\documentclass[11pt,a4paper]{article}
\usepackage[T1]{fontenc}
\usepackage[margin=1in]{geometry}
\usepackage{doi}
\usepackage{amsthm}
\theoremstyle{plain}
\newtheorem{theorem}{Theorem}[section]
\newtheorem{proposition}[theorem]{Proposition}
\newtheorem{lemma}[theorem]{Lemma}
\newtheorem{corollary}[theorem]{Corollary}
\newtheorem{assumption}[theorem]{Assumption}
\theoremstyle{definition}
\newtheorem{definition}[theorem]{Definition}
\newtheorem{example}[theorem]{Example}
\theoremstyle{remark}
\newtheorem{remark}[theorem]{Remark}

\AddToHook{env/lemma/begin}{\crefalias{theorem}{lemma}}
\AddToHook{env/proposition/begin}{\crefalias{theorem}{proposition}}
\AddToHook{env/corollary/begin}{\crefalias{theorem}{corollary}}
\AddToHook{env/definition/begin}{\crefalias{theorem}{definition}}
\AddToHook{env/assumption/begin}{\crefalias{theorem}{assumption}}
\AddToHook{env/example/begin}{\crefalias{theorem}{example}}
\AddToHook{env/remark/begin}{\crefalias{theorem}{remark}}

\usepackage{graphicx}
\usepackage{csquotes}
\usepackage{listings}
\usepackage[svgnames]{xcolor}
\usepackage{colortbl}
\usepackage{comment}
\usepackage{appendix}

\usepackage{amsmath}
\usepackage{amssymb}
\usepackage{physics}
\usepackage{braket}
\usepackage{mathrsfs}
\allowdisplaybreaks

\usepackage{booktabs}
\usepackage{multirow}
\usepackage{makecell}

\usepackage{algorithm}
\usepackage{algpseudocode}
\AddToHook{cmd/appendix/before}{}

\usepackage{tikz}
\usetikzlibrary{decorations.pathreplacing, positioning, matrix, calc, fit}
\tikzset{
    matrix/.style n args={2}{
        minimum height={#1 cm},
        minimum width={#2 cm},
        draw, fill=.!10,
        anchor=north west,
    },
    every label/.style={font=\Large},
    eqn/.style={font=\Large\boldmath, right}, 
}
\usetikzlibrary{backgrounds, arrows.meta}

\definecolor{darkgray}{rgb}{0.31,0.31,0.33}
\definecolor[named]{lipicsGray}{rgb}{0.31,0.31,0.33}
\definecolor[named]{lipicsBulletGray}{rgb}{0.60,0.60,0.61}
\definecolor[named]{lipicsLineGray}{rgb}{0.51,0.50,0.52}
\definecolor[named]{lipicsLightGray}{rgb}{0.85,0.85,0.86}
\definecolor[named]{lipicsYellow}{rgb}{0.99,0.78,0.07}

\algnewcommand{\IfThenElse}[3]{%
   \algorithmicif\ #1\ \algorithmicthen\ #2\ \algorithmicelse\ #3}
\algnewcommand{\GetsIfThenElse}[3]{%
   #1\ \algorithmicif\ #2 \algorithmicelse\ #3}
\algnewcommand{\IfThen}[2]{%
   \algorithmicif\ #1\ \algorithmicthen\ #2}

\newcommand{\F}{\mathbb{F}}

\newcommand{\used}{YES}
\newcommand{\notused}{NO}
\newcommand{\slwe}{\textsf{Search}-LWE}
\newcommand{\dlwe}{\textsf{Decision}-LWE}
\newcommand{\belwe}{\textsf{BinaryError}-LWE}
\newcommand{\bslwe}{\textsf{BinarySecret}-LWE}

\newcommand{\LT}{\ensuremath{\mathrm{LT}}}
\newcommand{\LM}{\ensuremath{\mathrm{LM}}}

\usepackage{hyperref} 
\usepackage{cleveref}
\usepackage{orcidlink}

\begin{document}
 
\title{A New Algebraic Algorithm for LWE}
\author{
  Luca Campa\orcidlink{0009-0004-6916-1918}, \quad Massimo Fumiani\orcidlink{0009-0002-4605-4314}, \quad Arnab Roy\orcidlink{0000-0002-3284-7076} \\[0.5em]
  \small University of Innsbruck, Innsbruck 6020, Austria \\
  \small \texttt{\{Luca.Campa,Massimo.Fumiani,Arnab.Roy\}@uibk.ac.at}
}
\date{}
\maketitle

\begin{abstract}
The Learning With Errors (LWE) problem, introduced by Regev in 2005, is central to modern cryptography and post-quantum security. The algorithms to solve the search version of the problem, \slwe\, can be broadly categorised into algebraic, combinatorial and lattice-based. 

In this work we propose a new algebraic algorithm for the \slwe\ problem. At a high level, the algorithm combines linear-algebraic techniques with S-polynomial-based methods from Gr{\"o}bner basis computation. We provide a direct complexity analysis of our algorithm, avoiding semi-regularity assumptions and complexity bounds derived from the degree of regularity. Our algorithm achieves a polynomial improvement in complexity over prior results that use Gr{\"o}bner basis methods to solve \slwe.
\end{abstract}

\section{Introduction}\label{sec:introduction}
The Learning With Errors (LWE) problem, introduced by Oded Regev~\cite{STOC:Regev05}, is a natural generalization of the learning parity with noise problem over finite fields.

\begin{definition}[LWE~\cite{EPRINT:AlbPlaSco15,STOC:Regev05}]\label{def:LWE}
    Let $n$, $q$ be positive integers, $\chi$ be a probability distribution on $\mathbb{Z}$ and $\mathbf{s}$ be a secret vector in $\mathbb{Z}_q^n$.
    We denote by $L_{\mathbf{s},\chi}$ the probability distribution on $\mathbb{Z}_q^n\times\mathbb{Z}_q$ obtained by choosing $\mathbf{a}\in\mathbb{Z}_q^n$ uniformly at random, choosing $e\in\mathbb{Z}$ according to $\chi$ and considering it in $\mathbb{Z}_q$, and returning $(\mathbf{a},b) = (\mathbf{a}, \langle \mathbf{a}, \mathbf{s} \rangle + e)\in\mathbb{Z}_q^n\times\mathbb{Z}_q$.
    \begin{itemize}
        \item \dlwe\ is the problem of deciding whether pairs $(\mathbf{a},b)\in\mathbb{Z}_q^n\times\mathbb{Z}_q$ are sampled according to $L_{\mathbf{s},\chi}$ or the uniform distribution on $\mathbb{Z}_q^n\times\mathbb{Z}_q$.
        \item \slwe\ is the problem of recovering $\mathbf{s}$ from $(\mathbf{a},b) = (\mathbf{a}, \langle \mathbf{a}, \mathbf{s} \rangle + e)\in\mathbb{Z}_q^n\times\mathbb{Z}_q$ sampled according to $L_{\mathbf{s},\chi}$.
    \end{itemize}
  \end{definition}

Since its introduction, the LWE problem has been investigated in the algorithms and cryptography literature. In particular, the \emph{Bounded Errors} case was explored in~\cite{EPRINT:ACFP14,ICALP:AroGe11,EC:NMSU25,EPRINT:Steiner24b}. Introduced by Arora and Ge~\cite{ICALP:AroGe11}, an \emph{algebraic algorithm} generates a system of error-free multivariate polynomial equations from the LWE samples and solves the system using a linearization technique. In~\cite{EPRINT:ACFP14}, the complexity of the Arora-Ge technique was improved by using Gr{\"o}bner basis methods to solve the polynomial system of equations. Recent analyses~\cite{EC:NMSU25,EPRINT:Steiner24b} have further improved the complexity of this approach.

The \emph{Lattice-based algorithms}, using primal and dual approaches that reduce LWE to BDD/uSVP or SIS and rely on lattice reduction techniques such as BKZ and hybrid methods~\cite{EC:Albrecht17,EPRINT:AlbPlaSco15,RSA:LinPei11}, are the most efficient algorithms for solving LWE.\@
\emph{Combinatorial algorithms}, notably BKW-based~\cite{DCC:ACFFP15} methods, reduce the problem dimension at the cost of increased noise and are mainly effective for instances with small secrets or large sample sizes. In~\cite{EPRINT:AlbPlaSco15}, a summary of a wide range of algorithms for the LWE problem is provided. 

\paragraph*{Hardness and Cryptographic Application.}
As established in the literature, the worst-case approximate GapSVP reduces to the average-case LWE \cite{STOC:BLPRS13,DBLP:conf/stoc/Peikert09,STOC:Regev05,DBLP:journals/jacm/Regev09}.
Such results position LWE as a central hardness assumption in lattice-based cryptography, leading to numerous cryptographic constructions, including post-quantum cryptographic standards.
Examples of such constructions include secure public-key encryption schemes resilient to both chosen-plaintext~\cite{RSA:LinPei11,DBLP:conf/crypto/PeikertVW08,STOC:Regev05} and chosen-ciphertext attacks~\cite{DBLP:conf/eurocrypt/MicciancioP12,DBLP:conf/stoc/Peikert09,DBLP:conf/stoc/PeikertW08}, oblivious transfer protocols~\cite{DBLP:conf/crypto/PeikertVW08}, identity-based encryption~\cite{DBLP:conf/eurocrypt/AgrawalBB10,DBLP:conf/crypto/AgrawalBB10,DBLP:conf/eurocrypt/CashHKP10,DBLP:conf/stoc/GentryPV08}, leakage-resilient cryptography~\cite{DBLP:conf/tcc/AkaviaGV09,DBLP:conf/crypto/ApplebaumCPS09,DBLP:conf/innovations/GoldwasserKPV10}, fully homomorphic encryption~\cite{DBLP:conf/crypto/Brakerski12,DBLP:conf/innovations/BrakerskiGV12,DBLP:conf/focs/BrakerskiV11}, among many others.
Structured variants, including Ring-LWE~\cite{EC:LyuPeiReg10} and Module-LWE~\cite{EPRINT:BraGenVai11}, improve efficiency while retaining worst-case hardness, with MLWE underpinning practical schemes such as CRYSTALS-Kyber~\cite{NISTPQC:CRYSTALS-KYBER22}.

\subsection{Our Result}
We propose a new and simple algebraic attack on \slwe\ with \emph{bounded errors}, where following standard practice in the literature~\cite{EPRINT:AlbPlaSco15}, errors are sampled from a discrete Gaussian distribution.

Our method differs fundamentally from existing approaches. Instead of directly applying Gr{\"o}bner basis methods, we employ \emph{a combination of diagonalization and S-polynomial generation} techniques. We carry out a direct complexity analysis of the proposed algorithm.

Unlike previous Gr{\"o}bner basis based analyses~\cite{EPRINT:ACFP14,EPRINT:Steiner24b}, our method does not rely on semi-regularity assumptions or the estimation of the degree of regularity for deriving the complexity bounds. Thus, our algorithm and its complexity analysis are not affected by the non-existence of semi-regular systems or the difficulty of precisely estimating the degree of regularity.

Despite the absence of the above-mentioned assumptions, the complexity of our algorithm shows a polynomial improvement (\Cref{tab:comparison}) over existing algebraic algorithms based on Gr{\"o}bner basis methods. Note that the requirement of a large number of samples for our algorithm remains the same as in existing algebraic algorithms. A detailed comparison with prior algebraic attacks is provided in Table~\ref{tab:comparison}.
Our main result is summarized in \Cref{thm:result}:

\begin{theorem}[Main Result]\label{thm:result}
Let $q$ be a prime and $2<\omega\leq3$ is the linear algebra constant. Consider a \slwe\ instance over $\mathbb{F}_q$ in dimension $n$, where the errors are drawn from a known support set of size $d$, such that $d=\mathcal{O}(\sqrt{n})$ and $d<q$. Given $m = \binom{n+d-1}{d}$ samples, our algorithm recovers the secret vector $\mathbf{s}$ in
\[
    \mathcal{O}\!\left(d\cdot\frac{n^{d\omega}}{{(d!)}^\omega}\right)
\]
field operations over $\mathbb{F}_q$, with high probability $p_{\mathrm{succ}}\simeq 1-\mathcal{O}(d/q)$.
\end{theorem}

For $d=2$, the above result is proven unconditionally; for $d>2$, it requires the mild condition~\ref{ass:generic}, which does not restrict the applicability of our algorithm as it is only used to lower bound its success probability.

In particular, with $m = \binom{n+d-1}{d}$ samples and an error support of size $d$, our complexity bound improves upon existing algebraic algorithms~\cite{EPRINT:ACFP14,EC:NMSU25,EPRINT:Steiner24b,ACISP:SunTibAbe20} for LWE.\@ More specifically, comparing with~\cite{EC:NMSU25}, the most efficient existing method that shares the same data complexity $m$, our algorithm improves the computational complexity by a factor of \(\sim n{\left(d!\right)}^{\omega}\). Note that, when $d = 2$, our algorithm also yields a polynomial speed-up by a factor of \(\sim n^{8 + \omega}\) w.r.t.~\cite{EPRINT:Steiner24b}.

Like the algorithm of~\cite{EC:NMSU25}, our algorithm is probabilistic and succeeds with high probability (see~\Cref{sec:probabilities2,sec:probabilitiesD}). The success guarantee of our algorithm, however, imposes a very different requirement on the field size, namely, the success probability of~\cite{EC:NMSU25} requires $q > md + 2$~\cite[Thm.~4]{EC:NMSU25}, a condition that is not met for practical parameters, e.g.~\cite{NISTPQC:CRYSTALS-KYBER22,FRODOKEM}. A (high) success probability of our algorithm requires only $q \ge 3d$. Our algorithm therefore improves not only the time complexity, but also the range of applicability, since it applies to practical instances.

\begin{table}[b!]
    \centering
    \renewcommand{\arraystretch}{1.5} 
    \begin{tabular}{
    >{\centering\arraybackslash}m{1.6cm}
    >{\centering\arraybackslash}m{1cm}
    >{\centering\arraybackslash}m{2.5cm}
    >{\centering\arraybackslash}m{3.8cm}
    >{\centering\arraybackslash}m{2.8cm}
    >{\centering\arraybackslash}m{1.3cm}
    }
    \hline\hline
Work & Error Size & No. Samples & Time Complexity & Semi-regular assumption? & Using $d_{reg}$? \\ 
\hline\hline
~\cite{EPRINT:ACFP14} & 2 & $\mathcal{O}(n\log\log n)$ & $\mathcal{O}\left(n^2\cdot 2^{\frac{\omega n \log\log\log n}{8\log\log n}}\right)$ & \used& \used\\
\hline
~\cite{EPRINT:ACFP14} & 2 & $c\cdot n$ & $\mathcal{O}\Big(n^2\cdot 2^{\omega n(1+\beta)H_2(\beta/(1+\beta))}\Big)$ & \used& \used\\
\hline
~\cite{ACISP:SunTibAbe20} & 2 & $c\cdot n^2$ & $n^{\mathcal{O}(1/c)}$ & \used& \used\\
\hline
~\cite{ACISP:SunTibAbe20} & 2 & $n^{1+\alpha}$ & $2^{\widetilde{\mathcal{O}}(m^{1-\alpha})}$ & \used& \used\\
\hline
~\cite{EPRINT:Steiner24b} & 2 & $\mathcal{O}({n^2})$ & $\mathcal{O}\left( n^2\cdot \binom{n+3}{3}^{\omega +2} \right)$ & \used& \used\\
\hline
This Work & $2$ & $\frac{n(n+1)}{2}$ & $\mathcal{O}\left(\frac{n^{2\omega}}{2^{\omega-1}}\right) $ &\cellcolor{lipicsLightGray} \notused&\cellcolor{lipicsLightGray} \notused\\
\hline\hline
~\cite{ICALP:AroGe11} & $d$ & $\mathcal{O}(\log(q)\cdot q\cdot n^d)$ & $\mathcal{O}\left(\log(q)\cdot q\cdot n^{\omega d}\right)$ &\cellcolor{lipicsLightGray} \notused&\cellcolor{lipicsLightGray} \notused\\ 
\hline
~\cite{EPRINT:Steiner24b} & $d$ & $m>n$ & $\mathcal{O}\Big( m\cdot (d-1)\cdot n\cdot 2^{\omega\cdot{(8d\ln(4)-1)}^{1/\ln(4)}\cdot n} \Big)$ & \cellcolor{lipicsLightGray}\notused& \used\\
\hline
~\cite{EPRINT:Steiner24b} & $d$ & $\mathcal{O}\Big( \binom{n+d-1}{d} \Big)$ & $\mathcal{O}\left(d^3\cdot c_d^{{(n-1)}^{1-1/\ln(4)}}\right)$ & \used& \used\\
\hline
~\cite{EC:NMSU25} & $d$ & $\binom{n+d-1}{d}$ & $\mathcal{O}\left(dn^{1+d\omega}\right)$ & \cellcolor{lipicsLightGray}\notused& \used\\
\hline
This Work & $d$ & $\binom{n+d-1}{d}$ & $\mathcal{O}\left(d\cdot \frac{n^{d\omega}}{{(d!)}^\omega}\right) $ &\cellcolor{lipicsLightGray} \notused&\cellcolor{lipicsLightGray}\notused\\
\hline\hline\\
    \end{tabular}
    \caption{Overview of the complexities (in field operations) of algebraic attacks on LWE with a secret $\mathbf{s}$ of length $n$, $m$ samples, and error entries drawn from a set of size $d$. $H_2(x) = -x \log x - (1-x) \log(1-x)$ is the binary entropy function, $\ln$ is the natural logarithm, $\omega$ is the linear algebra constant ($2 < \omega \le 3$), $c_d = 2^{(\omega+3)\cdot 2^{1/\ln(2)}\cdot {(2d-1)}^{1/\ln(4)}}$, and $\beta = c-0.5-\sqrt{c(c-1)}$. This table is an extended version of Table 2 in~\cite{EC:NMSU25}, with adapted notation.}\label{tab:comparison}
\end{table}

\subsection{Overview of Previous Algebraic Algorithms}

\paragraph*{Linearization Technique~\cite{ICALP:AroGe11}.}\label{subsection:linearization}
Each sample $(\mathbf{a}_i, b_i) = (\mathbf{a}_i, \langle\mathbf{a}_i, \mathbf{s}\rangle + e_i) \in \mathbb{F}_q^n \times \mathbb{F}_q$, for $1 \leq i \leq m$, yields a non-linear equation of degree $d$ in the $n$ components of the secret $\mathbf{s}$.
Given that a common choice for the error distribution is the discrete Gaussian with mean $0$ and standard deviation $\sigma$, denoted as $\mathcal{N}(0,\sigma)$, we have
\[
\mathbb{P}\big[ e \overset{\scriptscriptstyle\$}{\gets} \chi : |e| > t \cdot \sigma \big] \leq \frac{2}{t\sqrt{2\pi}} e^{-t^2/2}\in e^{\mathcal{O}({-t^2})}, \quad \textrm{for all } t>0.
\]
When representing elements of $\mathbb{F}_q$ as integers in $[-\lfloor q/2 \rfloor, \dots, \lfloor q/2 \rfloor]$, a sample from a Gaussian distribution takes values in the interval $[-t\cdot\sigma,\dots,t\cdot\sigma]$ of $\mathbb{F}_q$ with probability $1 - e^{\mathcal{O}({-t^2})}$.
Consequently, if $e \overset{\scriptscriptstyle \$}{\gets} \mathcal{N}(0,\sigma)$, then $f(e)=0$ for $f(x) = x \cdot \prod_{i=1}^{t\cdot\sigma} (x+i)(x-i)$ with probability at least $1 - e^{\mathcal{O}({-t^2})}$.
The degree of $f$ is $d := 2 \cdot t \sigma + 1 < q$.
The resulting polynomial system $\mathcal{F}_{\mathrm{LWE}}$ is solved using the linearization method, which involves replacing each monomial with a new linear variable and then solving the resulting linear system of equations.
\belwe\ is the variant of LWE introduced in~\cite{MicPei13}, in which $\mathbf{e} \in {\{0,1\}}^m$.

\paragraph*{Gr{\"o}bner Basis and Degree of Regularity.}
In~\cite{EPRINT:ACFP14}, the Gr{\"o}bner basis algorithm (a classical algebraic geometry method for solving systems of polynomial equations) was applied to the error-free polynomial system, achieving a significant improvement in runtime over the Arora-Ge algorithm. The complexity analysis of the Gr{\"o}bner basis algorithm is non-trivial and typically requires computing the \emph{degree of regularity} ($d_{reg}$, see \Cref{sec:dreg}) under the semi-regularity assumption. While computing $d_{reg}$ for a polynomial system is difficult, for a semi-regular system, one can estimate it. Assuming that the polynomial system (from LWE) is semi-regular, the complexity analysis (of the Gr{\"o}bner basis algorithm) in~\cite{EPRINT:ACFP14} provides estimates of the degree of regularity. Under the same assumption,~\cite{ACISP:SunTibAbe20} estimates $d_{reg}$ using the Hilbert polynomial. While most random systems of polynomials are semi-regular, Fr\"oberg showed~\cite{DBLP:journals/jsc/FrobergH94} that there are polynomial systems that can never be semi-regular, and~\cite{HODGES2017519} showed that (for certain parameter choices) semi-regular sequences can never exist. These limitations motivate complexity analyses that avoid explicit semi-regularity assumptions.
Recently,~\cite{EPRINT:Steiner24b} provided a refined complexity analysis using the solving degree and assuming the degree of regularity to be the same as the degree of the polynomial. While this result does not use the semi-regular assumption, it still relies on an assumption of the degree of regularity.
In~\cite{EC:NMSU25}, the $d_{reg}$ was shown to be $d$ with high probability~\cite[Lems.~5,6]{EC:NMSU25}.

We emphasize that our approach does not aim to improve existing bounds on the degree of regularity; instead, we avoid this notion altogether by providing a fine-grained complexity directly in terms of the LWE parameters.

\section{Preliminaries}\label{sec:background}
For $\alpha = (\alpha_1,\dots,\alpha_n) \in \mathbb{Z}_{\ge 0}^n$ we define $\mathbf{x}^{\alpha} := x_1^{\alpha_1}\cdots x_n^{\alpha_n}$. The total degree of $\mathbf{x}^{\alpha}$ is $|\alpha|$, where $|\alpha| := \sum_{i=1}^n \alpha_i$. For a polynomial $f \in \mathbb{F}_q[x_1,\dots,x_n]$ of total degree $d$, we adopt the following degree decomposition:
\[
f = f^{(d)} + f^{(d-1)} + \cdots + f^{(1)} + f^{(0)},
\]
where $f^{(k)}$ denotes the sum of all monomials in $f$ of total degree exactly $k$.
For a fixed monomial order $\prec$ on $\mathbb{F}_q[x_1,\dots,x_n]$, \LM$(f)$ and \LT$(f)$ denote, respectively, the leading monomial and leading term of $f$ with respect to $\prec$. The S-polynomial for $f, g \in \mathbb{F}_q[x_1,\dots,x_n]$, which is fundamental to the Gr\"obner basis method, is defined as
\begin{equation}\label{eq:spoly-expanded}
    S_{fg} := \frac{\mathbf{x}^\gamma}{\mathrm{LT}(f)} \cdot f - \frac{\mathbf{x}^\gamma}{\mathrm{LT}(g)} \cdot g,
\end{equation}
where $\mathbf{x}^\gamma$ denotes the least common multiple of $\mathrm{LM}(f)$ and $\mathrm{LM}(g)$. The Gr\"obner basis method is a well-known technique for solving systems of polynomials. For a complete introduction to the topic of Gr\"obner bases, we refer the reader to~\cite{Cox2015Ideals}.

\section{An Overview of Our Algorithm}\label{sec:idea}
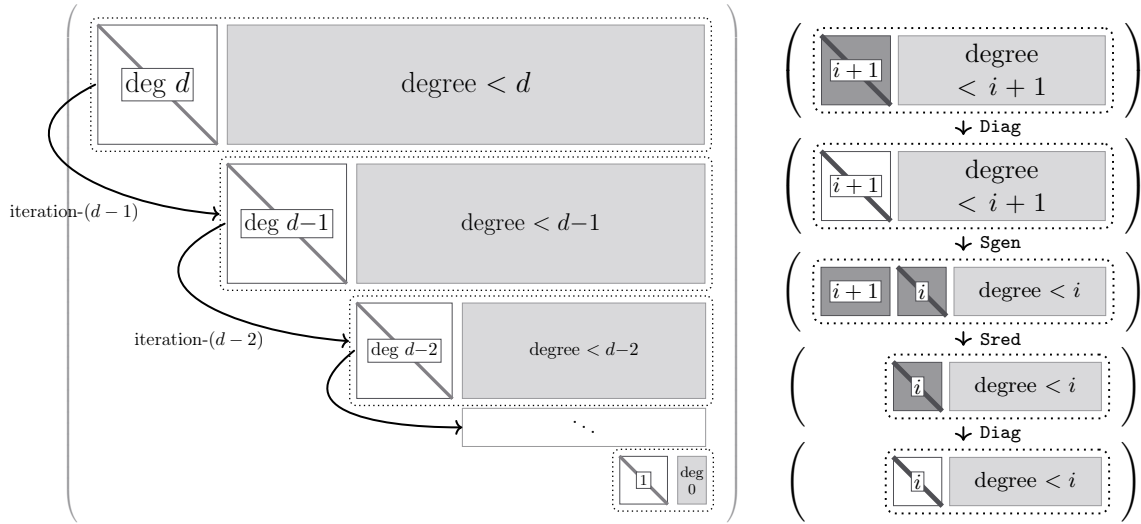
\begin{figure}[b!]
\centering
\resizebox{0.95\linewidth}{!}{
\begin{tabular}{cc}
  \resizebox{0.60\textwidth}{!}{\begin{tikzpicture}[
    every label/.append style={anchor=base, yshift=3pt},
    every left delimiter/.style={lipicsBulletGray},
    every right delimiter/.style={lipicsBulletGray},
    hugeblock/.style={
        inner sep=6pt,
        left delimiter={(},
        right delimiter={)}
    }
]
\node[lipicsGray, matrix={2.5}{2.5}, fill=none] (d) {};
\draw[lipicsLineGray, line width=0.70mm] (d.north west) -- (d.south east) -- cycle;
\node[font=\LARGE, fill=white, draw=lipicsGray, line width=0.4pt, inner sep=2pt] at (d.center) {deg $d$};
\node[lipicsBulletGray, matrix={2.5}{10}, fill=lipicsLightGray, right=2mm of d] (ld) {};
\node[font=\LARGE] at (ld.center) {degree $<d$};
\node[draw, dotted, thick, rounded corners, inner sep=4pt, fit=(d)(ld)] (boxd) {};
\node[lipicsGray, matrix={2.5}{2.5}, fill=none, below=4mm of ld.south west, anchor=north west] (d1) {};
\draw[lipicsLineGray, line width=0.70mm] (d1.north west) -- (d1.south east) -- cycle;
\node[font=\Large, fill=white, draw=lipicsGray, line width=0.4pt, inner sep=2pt] at (d1.center) {deg $d{-}1$};
\node[lipicsBulletGray, matrix={2.5}{7.3}, fill=lipicsLightGray, right=2mm of d1] (ld1) {};
\node[font=\Large] at (ld1.center) {degree $<d{-}1$};
\node[draw, dotted, thick, rounded corners, inner sep=4pt, fit=(d1)(ld1)] (boxd1) {};
\node[lipicsGray, matrix={2}{2}, fill=none, below=4mm of ld1.south west, anchor=north west] (d2) {};
\draw[lipicsLineGray, line width=0.70mm] (d2.north west) -- (d2.south east) -- cycle;
\node[font=\normalsize, fill=white, draw=lipicsGray, line width=0.4pt, inner sep=2pt] at (d2.center) {deg $d{-}2$};
\node[lipicsBulletGray, matrix={2}{5.1}, fill=lipicsLightGray, right=2mm of d2] (ld2) {};
\node[font=\normalsize] at (ld2.center) {degree $<d{-}2$};
\node[draw, dotted, thick, rounded corners, inner sep=4pt, fit=(d2)(ld2)] (boxd2) {};
\node[inner sep=0pt, draw=lipicsBulletGray, outer sep=0pt, matrix={0.8}{5.1}, fill=none, below=2mm of ld2, label=center:$\ddots$] (dots) {};
\node[lipicsGray, matrix={1}{1}, fill=none, below=2mm of dots.south east, xshift=-13mm] (d1block) {};
\draw[lipicsLineGray, line width=0.70mm] (d1block.north west) -- (d1block.south east) -- cycle;
\node[font=\footnotesize\bfseries, fill=white, draw=lipicsGray, line width=0.4pt, inner sep=2pt] at (d1block.center) {$1$};
\node[lipicsBulletGray, matrix={1}{0.6}, fill=lipicsLightGray, right=2mm of d1block] (smalld1) {};
\node[font=\footnotesize] at (smalld1.center) {\shortstack{deg \\ $0$}};
\node[draw, dotted, thick, rounded corners, inner sep=4pt, fit=(d1block)(smalld1)] (box1) {};
\node[hugeblock, fit=(boxd)(boxd1)(boxd2)(dots)(box1)] (big) {};
\draw[->, very thick, line cap=round, line join=round, overlay]
([xshift=3pt]boxd.west) .. controls +(-24mm,-12mm) and +(-25mm,2mm) .. node[below=6mm, xshift=1mm] {iteration-$(d-1)$} ([yshift=2mm]boxd1.west);
\draw[->, very thick, line cap=round, line join=round, overlay]
([xshift=3pt]boxd1.west) .. controls +(-24mm,-12mm) and +(-25mm,2mm) .. node[below=5mm] {iteration-$(d-2)$} ([yshift=2mm]boxd2.west);
\draw[->, very thick, line cap=round, line join=round, overlay]
([xshift=3pt]boxd2.west) .. controls +(-15mm,-12mm) and +(-20mm,0mm) .. (dots.west);
\end{tikzpicture}} &
  \resizebox{0.33\textwidth}{!}{\begin{tikzpicture}[
  matparens/.style={
    inner sep=3pt,
    left delimiter={(},
    right delimiter={)}
  },
  block/.style={draw=gray, thick, inner sep=2pt, minimum height=12mm, minimum width=12mm, font=\large\bfseries},
  bullet/.style={draw=gray, fill=gray!30, minimum width=12mm, minimum height=12mm},
  arrow/.style={->, thick, line cap=round},
  node distance=2mm
]
\node[matparens] (A1) {
    \begin{tikzpicture}
        \node[lipicsGray, matrix={1}{1}, fill=lipicsBulletGray] (R) {};
        \draw[lipicsGray, line width=0.70mm] (R.north west) -- (R.south east) -- cycle;
        \node[font=\footnotesize, fill=white, draw=lipicsGray, line width=0.4pt, inner sep=1pt] at (R.center) {$i+1$};
        \node[lipicsBulletGray, matrix={1}{3}, fill=lipicsLightGray, right=1mm of R] (T) {};
        \node[black, align=center, text width=21mm, anchor=center] at (T.center) {degree $<i+1$};
        \node[draw, dotted, thick, rounded corners, inner sep=3pt, fit=(R)(T)] (boxd) {};
    \end{tikzpicture}
  };
\node[matparens, below=of A1] (B2) {
    \begin{tikzpicture}
        \node[lipicsGray, matrix={1}{1}, fill=white] (P) {};
        \draw[lipicsGray, line width=0.70mm] (P.north west) -- (P.south east) -- cycle;
        \node[font=\footnotesize, fill=white, draw=lipicsGray, line width=0.4pt, inner sep=1pt, anchor=center] at (P.center) {$i+1$};
        \node[lipicsBulletGray, matrix={1}{3}, fill=lipicsLightGray, right=1mm of P] (L) {};
        \node[black, align=center, text width=21mm, anchor=center] at (L.center) {degree $<i+1$};
        \node[draw, dotted, thick, rounded corners, inner sep=3pt, fit=(P)(L)] (boxd) {};
    \end{tikzpicture}
};

\draw[arrow]
  (A1.south) -- (B2.north)
  node[midway, right=1mm, font=\scriptsize, align=center] {\texttt{Diag}};

\node[matparens, below=of B2] (A1b) {
\begin{tikzpicture}
    \node[lipicsGray, matrix={0.7}{1}, fill=lipicsBulletGray] (K) {};
    \node[font=\footnotesize, fill=white, draw=lipicsGray, line width=0.4pt, inner sep=1pt, anchor=center] at (K.center) {$i+1$};
    \node[lipicsGray, matrix={0.7}{0.7}, fill=lipicsBulletGray, right=1mm of K] (R) {};
    \draw[lipicsGray, line width=0.70mm] (R.north west) -- (R.south east) -- cycle;
    \node[font=\footnotesize, fill=white, draw=lipicsGray, line width=0.4pt, inner sep=1pt, anchor=center] at (R.center) {$i$};
    \node[lipicsBulletGray, matrix={0.7}{2.2}, fill=lipicsLightGray, right=1mm of R] (T) {};
    \node[black, align=center, text width=16mm, anchor=center, font=\footnotesize] at (T.center) {degree $<i$};
    \node[draw, dotted, thick, rounded corners, inner sep=3pt, fit=(K)(R)(T)] (boxd) {};
\end{tikzpicture}
};

\draw[arrow] (B2.south) -- (A1b.north) node[midway, right=1mm, font=\scriptsize, align=center] {\texttt{Sgen}};

\node[matparens, below=of A1b] (B3) {
\begin{tikzpicture}
    \node[matrix={0.7}{1}, draw=none, fill=none] (K) {};
    \node[lipicsGray, matrix={0.7}{0.7}, fill=lipicsBulletGray, right=1mm of K] (R) {};
    \draw[lipicsGray, line width=0.70mm] (R.north west) -- (R.south east) -- cycle;
    \node[font=\footnotesize, fill=white, draw=lipicsGray, line width=0.4pt, inner sep=1pt, anchor=center] at (R.center) {$i$};
    \node[lipicsBulletGray, matrix={0.7}{2.2}, fill=lipicsLightGray, right=1mm of R] (T) {};
    \node[black, align=center, text width=16mm, anchor=center, font=\footnotesize] at (T.center) {degree $<i$};
    \node[draw, dotted, thick, rounded corners, inner sep=3pt, fit=(R)(T)] (boxd) {};
\end{tikzpicture}
};

\draw[arrow] (A1b.south) -- (B3.north) node[midway, right=1mm, font=\scriptsize, align=center] {\texttt{Sred}};

\node[matparens, below=of B3] (AA1) {
    \begin{tikzpicture}
        \node[matrix={0.7}{1}, draw=none, fill=none] (K) {};
        \node[lipicsGray, matrix={0.7}{0.7}, fill=white, right=1mm of K] (P) {};
        \draw[lipicsGray, line width=0.70mm] (P.north west) -- (P.south east) -- cycle;
        \node[font=\footnotesize, fill=white, draw=lipicsGray, line width=0.4pt, inner sep=1pt, anchor=center] at (P.center) {$i$};
        \node[lipicsBulletGray, matrix={0.7}{2.2}, fill=lipicsLightGray, right=1mm of P] (L) {};
        \node[font=\footnotesize, black, align=center, text width=16mm, anchor=center] at (L.center) {degree $<i$};
        \node[draw, dotted, thick, rounded corners, inner sep=3pt, fit=(P)(L)] (boxd) {};
    \end{tikzpicture}
};

\draw[arrow]
  (B3.south) -- (AA1.north)
  node[midway, right=1mm, font=\scriptsize, align=center] {\texttt{Diag}};
\end{tikzpicture}}
\end{tabular}
}
\caption{\textbf{{Left:}} Block-wise representation of the algorithm. Each row block (enclosed by dashed lines) corresponds to polynomials of total degree $i$, from $d$ down to $1$. The left part of each rectangular matrix shows the diagonalized leading submatrix of degree-$i$ monomials, while the right part contains lower-degree terms. \textbf{{Right:}} Step-wise illustration of one iteration of the algorithm, showing \texttt{Sgen}, \texttt{Sred}, and \texttt{Diag}. Gray backgrounds represent dense matrices, while white backgrounds indicate zero entries.}\label{fig:bigmatrix}
\end{figure}

In the polynomial system $\mathcal{F}_{\mathrm{LWE}}$ (stemming from \slwe), derived from $m = \binom{n+d-1}{d}$ samples, we identify each polynomial with its coefficient vector. The terms of each polynomial are ordered according to a fixed \emph{graded} monomial ordering by total degree. 
This yields a matrix $\mathbf{M} \in \mathbb{F}_q^{m \times u}$ whose rows correspond to the input polynomials, where $u = \binom{n+d}{d}$.
With high probability, these rows are linearly independent and moreover the leading principal $m \times m$ submatrix of $\mathbf{M}$ is invertible (see \Cref{sec:probabilities2,sec:probabilitiesD}); diagonalizing it produces a set of polynomials with monic leading terms of degree $d$ (outlined by the first square structure in \Cref{fig:bigmatrix}).
The remaining columns form a dense block containing lower-degree monomials. Each polynomial in this representation has a unique term of degree $d$ corresponding to the leading term, and additional terms of lower degree. The leading terms are all distinct and, due to the linear independence of the rows, they cover all monomials of degree $d$.

From this system, we generate a selected set of S-polynomials, each having total degree $ \leq d$. Next, we reduce this set of S-polynomials by the current diagonalized system. Note that this process \emph{does not require division}, as it consists of linear combinations of the diagonalized polynomials. Each reduced S-polynomial has total degree exactly $d-1$.

The coefficients of these reduced S-polynomials are then assembled into a new coefficient matrix.
Provided that a sufficient number of such polynomials are generated, the leading principal submatrix of this matrix is invertible with overwhelming probability (see \Cref{sec:probabilities2,sec:probabilitiesD}) and can be diagonalized in the same manner as before.
This yields a new diagonalized system whose leading terms are monic and cover all monomials of degree $d-1$. Then, we compute a set of S-polynomials and their reduction as described above. 

This procedure is applied recursively, decreasing the degree at each iteration. Finally, the process terminates at degree $1$, where the resulting system consists of linear polynomials of the form $(x_j - s_j)$ for $1\le j\le n$, directly revealing the secret vector $\mathbf{s}$.

\subsection{Iterations (formal)}
The main algorithm (\Cref{alg:alg}) proceeds in $d-1$ iterations. At iteration $i$ (running from $i = d-1$ down to $1$), the algorithm processes polynomials of total degree $i+1$ by performing the following three operations:  

\begin{enumerate}
    \item \textbf{S-polynomial generation (\texttt{Sgen}).}
    We generate $\binom{n+i-1}{i}$ S-polynomials, ensuring that the set contains exactly as many equations as there are monomials of degree $i$.
    An S-polynomial between two polynomials $f$ and $g$ is computed if the greatest common divisor of their (distinct) leading terms has degree $i$.
    We recall that, at this stage, \LT$(f)$ and \LT$(g)$ have degree $i+1$.
    The resulting S-polynomial has the form  
    \begin{equation}\label{eq:spoly}
        S_{fg} = x_k\cdot f - x_\ell \cdot g
    \end{equation}
    for suitable indices $k,\ell \in \{1,\dots,n\}$.
    Consequently, $S_{fg}$ consists solely of monomials of degree at most $i+1$.
    
    \item \textbf{S-polynomial reduction (\texttt{Sred}).}
    Exploiting the diagonal structure of the starting set, the S-polynomials are reduced strictly via scalar linear combinations, without explicit polynomial division. This cancels all monomials of degree $i+1$, leaving terms of degree at most $i$.

    \item \textbf{Diagonalization (\texttt{Diag}).}
    We compute the Reduced Row Echelon Form (RREF) of the coefficient matrix formed strictly by these reduced S-\hspace{0pt}polynomials, explicitly diagonalizing its leading principal submatrix of size $\binom{n+i-1}{i}$.
\end{enumerate}

\begin{algorithm}[t!]
\caption{(Main Algorithm)}\label{alg:alg}
\begin{algorithmic}[1]
\Statex{\hspace*{-\algorithmicindent} \textbf{Input:} $\mathcal{P}_d$: initial set of samples, $n \geq 1$: secret's dimension, $d$: size of the error support set}
\Statex{\hspace*{-\algorithmicindent} \textbf{Output:} $\mathcal{D}_1$: set of $n$ linear polynomials $\{ x_j - s_j: 1 \le j \le n \}$.}
\State{\texttt{SET = Diag($\mathcal{P}_d, n, d$)}} \Comment{Initial diagonalization of the degree-$d$ samples: $\mathcal{D}_d$}
\For{$i = d-1, d-2, \dots, 1$}
    \State{\texttt{STEP1 = Sgen(SET,$n, i$)}} \Comment{S-polynomial generation step: $\mathcal{P}_i$}
    \State{\texttt{STEP2 = Sred(STEP1,SET,$n, i$)}} \Comment{S-polynomial reduction step: $\mathcal{Q}_i$}
    \State{\texttt{SET = Diag(STEP2,$n, i$)}} \Comment{Diagonalization step: $\mathcal{D}_i$}
\EndFor{}
\State{\Return{\texttt{SET}}} \Comment{Return the reduced set of linear polynomials: $\mathcal{D}_1$}
\end{algorithmic}
\end{algorithm}

\section{Algorithm and Correctness}\label{sec:code}
Let $\mathcal{Q}_d$ denote the initial sample set of polynomials, and let $\mathcal{D}_d$ denote its corresponding diagonalized form, which serves as the starting point of the algorithm. 
For each subsequent iteration $i$ (from $d-1$ down to $1$), we define the following sets: let $\mathcal{P}_i$ denote the set of S-polynomials generated from $\mathcal{D}_{i+1}$, let $\mathcal{Q}_i$ denote the intermediate set obtained after reducing $\mathcal{P}_i$ (\texttt{Sred}), and let $\mathcal{D}_i$ denote the final set obtained after diagonalization. At each iteration $i$, these sets share the exact same cardinality:
$|\mathcal{P}_i| = |\mathcal{Q}_i| = |\mathcal{D}_i| = \binom{n+i-1}{i}$.

\Cref{alg:Sgen} details the generation procedure (\texttt{Sgen}), explicitly constructing the set of S-polynomials $\mathcal{P}_i$ from the diagonalized set $\mathcal{D}_{i+1}$. This process specifies the selection of admissible pairs $(f,g)$ (\Cref{def:admissible-pair}) to compute the polynomials $S_{fg}$ of degree $i+1$.

Subsequently, \Cref{alg:Sred} details the reduction step (\texttt{Sred}), which reduces the newly generated system $\mathcal{P}_i$ with respect to the diagonal set $\mathcal{D}_{i+1}$. By exploiting the pre-existing diagonal structure of $\mathcal{D}_{i+1}$, this reduction is achieved entirely via linear combinations, eliminating all monomials of degree $i+1$ to yield the reduced set $\mathcal{Q}_i$.

Finally, the diagonalization step (\texttt{Diag}) transforms the polynomial set $\mathcal{Q}_i$ into its required canonical diagonal form $\mathcal{D}_i$. Because this transformation corresponds to computing the RREF of the coefficient matrix associated with $\mathcal{Q}_i$, we omit a dedicated algorithm for this standard linear algebra procedure.

\subsection{S-polynomial generation}\label{sec:onspoly}
\begin{definition}[Admissible Pair]\label{def:admissible-pair}
    Given two tuples of integers \(\alpha, \beta \in \mathbb{Z}_{\ge 0}^n\) such that \(|\alpha|=|\beta|=t\), the pair \((\alpha, \beta)\) is called an \emph{admissible pair} if and only if there exist two distinct standard basis vectors \(e_j, e_k \in \mathbb{Z}_{\ge 0}^n\) with \(j \neq k\) such that \(\alpha_j > 0\) and \(\beta = \alpha - e_j + e_k\). Equivalently, we say that two polynomials \(f\) and \(g\) with \(\LM(f) = \mathbf{x}^\alpha\) and \(\LM(g) = \mathbf{x}^\beta\) form an \emph{admissible pair} of polynomials if and only if the pair \((\alpha, \beta)\) is admissible. In other words, if \(|\alpha|=|\beta|=t\), then \(\deg(\gcd(\LM(f),\LM(g))) = t-1\).
\end{definition}

\begin{algorithm}[b]
\caption{(\texttt{NextAdmissiblePair}) --- Next admissible pair generation}\label{alg:next-admissible-pair}
\begin{algorithmic}[1]
\Statex{\hspace*{-\algorithmicindent} \textbf{Input:} $\ell_j$: $n$-tuple, last component of previous pair, $n \ge 1$: dimension, $2 \le i\le d$: current degree}
\Statex{\hspace*{-\algorithmicindent} \textbf{Output:} $(\ell_{j+1}, \ell_{j+2})$: next admissible pair}
\Statex{\hspace*{-\algorithmicindent} \textbf{Uses:} \textsc{Successor}: get next combination in Gray code order}\Comment{See Appendix~\ref{sec:successor}}
    \State{$\ell_{j+1} \gets$ \textsc{Successor}($\ell_j$, $n$, $i$)} 
\State{$\ell_{j+2} \gets$ \textsc{Successor}($\ell_{j+1}$, $n$, $i$)} 
\State{\Return{$(\ell_{j+1}, \ell_{j+2})$}}
\end{algorithmic}
\end{algorithm}

\begin{algorithm}[t]
\caption{(\texttt{Sgen}) --- S-polynomial generation}\label{alg:Sgen}
\begin{algorithmic}[1]
\Statex{\hspace*{-\algorithmicindent} \textbf{Input:} $\mathcal{D}_{i+1}$: diagonal polynomial set, $n \ge 1$: dimension, $1\le i\le d-1$: current degree}
\Statex{\hspace*{-\algorithmicindent} \textbf{Output:} $\mathcal{P}_i$: set of reduced S-polynomials of degree $i$}

\State{$\mathcal{P}_i \gets \emptyset$}\Comment{Initialize the set of S-polynomials}
\State{PrevCombination $\gets \varnothing$}
\For{$k\gets1$ to $\binom{n+i-1}{i}$ by $1$}
    \State{$(\ell_{2k-1}, \ell_{2k}) \gets$ \texttt{NextAdmissiblePair}(PrevCombination, $n$, $i+1$)} 
    \State{PrevCombination $\gets \ell_{2k}$}
    \State{Select $(f,g) \in \mathcal{D}_{i+1}$ such that \(\LM(f) = \mathbf{x}^{{\ell}_{2k-1}}\) and \(\LM(g) = \mathbf{x}^{\ell_{2k}}\)} 
    \State{Compute $S_{fg}$} \Comment{See \Cref{eq:spoly}}
    \State{$\mathcal{P}_i \gets \mathcal{P}_i \cup \{S_{fg}\}$}
\EndFor{}
\State{\Return{$\mathcal{P}_i$}}
\end{algorithmic}
\end{algorithm}

Consider the iteration $i$. Note that, given a monomial $m$ of degree \(i+1\), there exists a unique polynomial $f \in \mathcal{D}_{i+1}$ such that \(\LM(f) = m\), due to the diagonal form of $\mathcal{D}_{i+1}$. Hence, given two monomials $m_1$ and $m_2$ that form an admissible pair, the corresponding admissible pair of polynomials can be uniquely determined. Let $f$ and $g$ be those two polynomials; the resulting S-polynomial \(S_{fg}\) admits the simplified form given in \Cref{eq:spoly}.

Since they form an admissible pair, the least common multiple of their leading terms has degree $i+2$.
Consequently, recalling that $\deg(\mathrm{LM}(f))=\deg(\mathrm{LM}(g))=i+1$, both $\frac{\mathbf{x}^\gamma}{\mathrm{LT}(f)}$ and $\frac{\mathbf{x}^\gamma}{\mathrm{LT}(g)}$ (in \Cref{eq:spoly-expanded}) are monomials of degree one. Using the notation of Section~\ref{sec:background}, we decompose
\[
    f = f^{(i+1)} + f^{(i)} + \cdots + f^{(1)} + f^{(0)}.
\] 
Each polynomial in $\mathcal{D}_{i+1}$ contains a unique monomial of degree $i+1$, namely $\mathrm{LM}(f)$.
Multiplication by a variable $x_k$ therefore yields
\[
    x_k \cdot f = f^{(i+2)} + f^{(i+1)} + \cdots + f^{(1)},
\]
where $f^{(i+2)}$ consists of a single monomial, namely $x_k \cdot \mathrm{LM}(f)$.
As a result, when forming the S-polynomial of two such polynomials, the degree-$(i+2)$ terms are identical by construction and therefore cancel out upon subtraction, and we obtain
\[
    S_{fg} = {S_{fg}}^{(i+1)} + {S_{fg}}^{(i)} + \cdots + {S_{fg}}^{(2)} + {S_{fg}}^{(1)}.
\]

The top-degree component, ${S_{fg}}^{(i+1)}$, can then be systematically eliminated by reduction with respect to $\mathcal{D}_{i+1}$. As detailed in Section~\ref{sec:onspolyred}, this reduction is performed exclusively via linear operations, yielding the fully reduced S-polynomial of the form:
\[
    S_{fg} = {S_{fg}}^{(i)} + {S_{fg}}^{(i-1)} + \cdots + {S_{fg}}^{(0)}.
\]

\paragraph*{Existence and construction of admissible pairs.}
\Cref{lemma:spoly_number} guarantees the existence of a sufficient number of admissible pairs to generate all required S-polynomials, while preventing the formation of redundant pairs.

\begin{lemma}[Feasibility condition]\label{lemma:spoly_number}
Let $1 \le i < d-1$, and let $\mathcal{D}_{i+1}$ be the set of diagonal polynomials of degree $i+1$, with $|\mathcal{D}_{i+1}| = \binom{n+i}{i+1}$.
Then there exists a collection of at least $\binom{n+i-1}{i}$ pairwise disjoint admissible pairs of elements of $\mathcal{D}_{i+1}$.
\end{lemma}

\begin{proof}
In order to prove this result, we first establish (1.) a necessary condition for the existence of a sufficient number of admissible pairs, and then show (2.) that this condition is also sufficient by providing an explicit construction of such pairs.

\begin{enumerate}
    \item To generate at least \(\binom{n+i-1}{i}\) \emph{pairwise disjoint} pairs, at least \(2\binom{n+i-1}{i}\) polynomials are required. 
    Hence, a necessary condition is
    \[
    \begin{split}
       \binom{n + i}{i+1} &\ge 2\binom{n+i-1}{i} \iff \frac{(n+i)!}{(n-1)!(i+1)!} \ge 2\frac{(n+i-1)!}{(n-1)!(i)!}\\
       \iff \frac{n+i}{i+1} &\ge 2 \iff n+i \ge 2(i+1) \iff  n \ge i+2 \iff n > i+1,
    \end{split}
    \] where, in our setting, \(i+1\) is at most \(d\). As a result, if \(n > d\), we can generate a sufficient number of distinct pairs. 

    \item Let $1 \le i < n-1$ and let $\mathcal{D}_{i+1}$ be the set of diagonal polynomials of degree $i+1$.
    We define $\mathcal{A} := \{ \alpha \in \mathbb{Z}_{\ge 0}^n \colon \exists f \in \mathcal{D}_{i+1} \text{ with } \LM(f) = \mathbf{x}^\alpha \}$, i.e., the set of all multi-indices corresponding to the leading monomials of polynomials in $\mathcal{D}_{i+1}$.
    Following~\cite{DBLP:journals/ejc/RuskeyS96}, we observe that $\mathcal{A}$ can be viewed as the set of all combinations of $n$ elements chosen $i+1$ at a time, or as the placement of $i+1$ identical balls into $n$ distinct boxes, with each box containing at most $i+1$ balls. From this perspective, the elements of $\mathcal{A}$ can be ordered so that consecutive elements form admissible pairs. Such an ordering is known as a \emph{Gray code for combinations of a multiset}~\cite{DBLP:journals/ejc/RuskeyS96}.  
    This Gray code can be interpreted as a Hamiltonian path in the graph whose vertices are the elements of $\mathcal{A}$, with edges connecting pairs of vertices if and only if they form an admissible pair. By definition, a Hamiltonian path covers all vertices of the graph, which allows us to generate exactly $\left\lfloor {\binom{n+i}{i+1}}/{2} \right\rfloor$ pairwise disjoint admissible pairs. Denoting the elements along the path by $\ell_j$, $1 \le j \le \binom{n+i}{i+1}$, we may form the pairs $(\ell_{2k-1}, \ell_{2k})$ for $1 \le k \le \left\lfloor {\binom{n+i}{i+1}}/{2} \right\rfloor$.
    This ensures that each polynomial appears in at most one pair and that all pairs are admissible.
    Hence, as shown in (1.), when $n>i+1$ we can generate at least $\binom{n+i-1}{i}$ pairwise disjoint \emph{admissible} pairs.
\end{enumerate}
\end{proof} 

``\emph{Gray codes for combinations of a multiset}''~\cite{DBLP:journals/ejc/RuskeyS96} admit a constant-time successor computation algorithm~\cite{DBLP:journals/cj/Takaoka99,DBLP:journals/corr/Takaoka15a}. This successor procedure can be used to efficiently compute the Hamiltonian path in the adjacency graph described in the proof of \Cref{lemma:spoly_number}. Consequently, admissible pairs can be generated incrementally and on-demand during the S-polynomial generation phase, without increasing the asymptotic complexity of~\Cref{alg:Sgen}.
For the sake of completeness, the pseudocode for the \textsc{Successor} procedure is provided in \Cref{alg:successor} (\Cref{sec:successor}).

\subsection{S-polynomial reduction}\label{sec:onspolyred}

\begin{algorithm}[b]
\caption{(\texttt{Sred}) --- S-polynomial reduction}\label{alg:Sred}
\begin{algorithmic}[1]
\Statex{\hspace*{-\algorithmicindent} \textbf{Input:} $\mathcal{P}_i$: set of S-polynomials, $\mathcal{D}_{i+1}$: diagonal polynomial set, $n \ge 1$: dimension, $1 \le i \le d-1$: current degree}
\Statex{\hspace*{-\algorithmicindent} \textbf{Output:} $\mathcal{Q}_i$: set of reduced polynomials w.r.t.\ $\mathcal{D}_{i+1}$}

\State{Let $(I \ A)$ be the coefficient matrix of $\mathcal{D}_{i+1}$} \Comment{$I$ corresponds to deg $i+1$}
\State{Let $(B \ C)$ be the coefficient matrix of $\mathcal{P}_i$} \Comment{Columns aligned with $(I \ A)$}
\State{$M_{\mathcal{Q}} \gets C - B \times A$} \Comment{Annihilate submatrix $B$ via Schur complement}
\State{Let $\mathcal{Q}_i$ be the polynomial set derived from $M_{\mathcal{Q}}$} \Comment{Map rows back to polynomials}
\State{\Return{$\mathcal{Q}_i$}}
\end{algorithmic}
\end{algorithm}

Once the set of S-polynomials $\mathcal{P}_i$ is generated, it must be reduced with respect to $\mathcal{D}_{i+1}$ to eliminate all monomials of degree $i+1$. Because $\mathcal{D}_{i+1}$ is strictly diagonalized, this step can entirely bypass the polynomial division typically used in standard Gr{\"o}bner basis techniques and be executed \emph{exclusively via linear operations}.

Let the block coefficient matrix associated with the joint system of $\mathcal{D}_{i+1}$ and $\mathcal{P}_i$ be defined as:
\[
M = \begin{pmatrix} I & A \\ B & C \end{pmatrix},
\]
where the top row-block $(I \ A)$ corresponds to $\mathcal{D}_{i+1}$ (with the identity matrix $I$ representing the isolated leading monomials of degree $i+1$) and the bottom row-block $(B \ C)$ corresponds to the newly generated $\mathcal{P}_i$. The matrices $A$ and $C$ encode the coefficients of the lower-degree terms. 

Reducing $\mathcal{P}_i$ with respect to $\mathcal{D}_{i+1}$ is algebraically equivalent to performing block Gaussian elimination to annihilate the submatrix $B$. This operation naturally yields the Schur complement~\cite{golub13} $C - BA$, which corresponds exactly to the coefficient matrix of the reduced polynomial set $\mathcal{Q}_i$. \Cref{alg:Sred} formalizes this matrix-algebraic reduction.

\subsection{Correctness}\label{sec:correct}
Throughout this section, we assume that the error terms satisfy the prescribed bound. Thus, the secret $\mathbf{s}$ is a common zero of all polynomials in the initial set of samples $\mathcal{P}_d$.

\begin{proposition}[Preservation of the LWE secret]\label{prop:preservation}
At each iteration $i$ of \Cref{alg:alg}, where $1 \le i \le d-1$, the secret $\mathbf{s}$ is a common zero of all polynomials in $\mathcal{D}_i$.
\end{proposition}

\begin{proof}
\textbf{{Base case.}}
$\mathcal{D}_d$ is obtained from $\mathcal{P}_d$ via diagonalization (\texttt{Diag}), which consists of linear combinations of polynomials in $\mathcal{P}_d$.  
Since $\mathbf{s}$ is a common zero of all polynomials in $\mathcal{P}_d$, it is also a zero of all polynomials in $\mathcal{D}_d$.

We proceed by induction on $i$, decreasing from $i=d-1$ to $i=1$.

\noindent
\textbf{{Inductive step.}}  
Assume that, at iteration $i$, $\mathbf{s}$ is a common zero of all polynomials in $\mathcal{D}_{i+1}$.
The following steps are performed:
\begin{enumerate}
    \item \emph{S-polynomial generation.} Let $f,g\in\mathcal{D}_{i+1}$. By the inductive hypothesis, they vanish at $\mathbf{s}$. By definition (\Cref{eq:spoly-expanded}), their S-polynomial also vanishes at $\mathbf{s}$.

    \item \emph{S-polynomial reduction, and diagonalization.} We showed in \Cref{sec:onspoly} that reductions do not involve polynomial divisions.
    Therefore, these steps only consist of linear combinations of previously
    generated polynomials. Since all such polynomials vanish at
    $\mathbf{s}$, the resulting polynomials also vanish at $\mathbf{s}$.
\end{enumerate}
Therefore, $\mathbf{s}$ is a common zero of all polynomials in $\mathcal{D}_i$, for all $1 \le i \le d$.
\end{proof}

By \Cref{prop:preservation}, $\mathbf{s}$ remains a common zero of the polynomials throughout all iterations.
In particular, at the last iteration, all polynomials in $\mathcal{D}_1$ are linear, monic, and of degree $1$.
Each equation, due to the diagonalization process, is of the form $x_i - s_i$ for $i \in \{1,\dots,n\}$, and the secret is uniquely determined as $\mathbf{s} = (s_1, \dots, s_n)$.

\begin{proposition}[Termination]\label{prop:termination}
\Cref{alg:alg} terminates after a finite number of iterations.
\end{proposition}

\begin{proof}
At each iteration of the algorithm, the total degree is decreased by $1$. Thus, the algorithm stops after $d-1$ iterations, yielding linear equations. 
\end{proof}

\section{Computational Complexity}\label{sec:complexity}
We now establish the overall computational complexity of \Cref{alg:alg} by analyzing the individual costs of the $\texttt{Sgen}$, $\texttt{Sred}$, and $\texttt{Diag}$ procedures. Our analysis strictly counts the number of field multiplications. Field additions are of the same asymptotic order and thus do not affect the overall complexity. 

Before proceeding with the detailed analysis, we prove the following algebraic relation, which will be used in the subsequent proofs to easily estimate and isolate the dominant term in our complexity calculations.

\begin{lemma}\label{lem:complexity_bound}
Let $n, h$ be positive integers such that $h = \mathcal{O}(\sqrt{n})$, and let $c \ge 1$ be a constant exponent. The following asymptotic bound holds:
\[
\sum_{i=1}^{h} \mathcal{O}\left( \binom{n+i-1}{i}^c \right) = \mathcal{O}\left( h \cdot \frac{n^{hc}}{{(h!)}^c} \right).
\]
\end{lemma}

\begin{proof}
Since $h = \mathcal{O}(\sqrt{n})$, there exists a constant $\gamma>0$ such that $h\le\gamma\sqrt{n}$. For a given $i \le h$, expanding the factorial yields
\[
\binom{n+i-1}{i} = \frac{n(n+1)\cdots(n+i-1)}{i!} = \frac{n^i}{i!}\prod_{j=0}^{i-1}\Bigl(1+\frac{j}{n}\Bigr).
\]
Using $1+x\le e^x$, together with $i \le h\le\gamma\sqrt{n}$, we obtain
\[
\prod_{j=0}^{i-1}\Bigl(1+\frac{j}{n}\Bigr)\le \exp({\sum_{j=0}^{i-1} \frac{j}{n}}) = \exp(\frac{i(i-1)}{2n}) \le e^{\gamma^2/2}.
\]
Note that this is where the hypothesis $h = \mathcal{O}(\sqrt{n})$ is needed: without it, the exponential factor would not be bounded by a constant. Hence $\binom{n+i-1}{i}\le e^{\gamma^2/2}\, n^i/i!$, where the constant $e^{\gamma^2/2}$ does not depend on $i$, yielding:
\[
\sum_{i=1}^{h} \mathcal{O}\left( \binom{n+i-1}{i}^c \right) = \sum_{i=1}^{h} \mathcal{O}\left( \frac{n^{ic}}{{(i!)}^c} \right).
\]
To rigorously bound this finite sum, we factor out the maximal term at $i=h$:
\begin{equation}\label{eq:parenthesis}
    \sum_{i=1}^{h} \frac{n^{ic}}{{(i!)}^c} = \frac{n^{hc}}{{(h!)}^c} \left( 1 + {\left(\frac{h}{n}\right)}^c + {\left(\frac{h(h-1)}{n^2}\right)}^c + \cdots + {\left(\frac{h!}{n^{h-1}}\right)}^c \right).
\end{equation}
Observe that the $k$-th term inside the parentheses (for $1 \le k \le h-1$) satisfies
\[{\left( \frac{h(h-1)\cdots(h-k+1)}{n^k} \right)}^c\le{\left(\frac{h^k}{n^k}\right)}^c = \alpha^k,\]
where $\alpha = {(h/n)}^c$, since each of the $k$ factors $h,h-1,\dots,h-k+1$ is at most $h$. Hence the term in patenthesis in~\Cref{eq:parenthesis} is bounded by the finite geometric series $\sum_{k=0}^{h-1} \alpha^k$, yielding
\[
\sum_{i=1}^{h} \frac{n^{ic}}{{(i!)}^c} \le \frac{n^{hc}}{{(h!)}^c}\, \sum_{k=0}^{h-1} \alpha^k.
\]
From $h = \mathcal{O}(\sqrt{n})$, we have $h/n < 1$. Since the exponent $c \ge 1$, it follows that $\alpha < 1$. Because $\alpha$ is strictly positive and bounded by $1$, every term $\alpha^k$ in the series is at most $1$. Since the series consists of exactly $h$ terms, the entire sum is at most $h \cdot 1 = h$. Therefore, the sum of all lower-degree terms introduces at most a multiplicative factor of $h$, yielding the final asymptotic bound of $\mathcal{O}\left( h \cdot \frac{n^{hc}}{{(h!)}^c} \right)$.
\end{proof}

\begin{theorem}\label{thm:iterationCost}
    In \Cref{alg:alg}, the computational complexities of the \textup{\texttt{Sgen}}, \textup{\texttt{Sred}}, and \textup{\texttt{Diag}} procedures across all their respective iterations are bounded as follows:
    \begin{enumerate}
        \item $\mathscr{C}_{\textup{\texttt{Sgen}}} = \mathcal{O}\left((d-1)\cdot \dfrac{n^{2(d-1)}}{{((d-1)!)}^2} \right)$,
        \item $\mathscr{C}_{\textup{\texttt{Sred}}} = \mathcal{O}\left((d-1)\cdot \dfrac{n^{1 + \omega(d-1)}}{{((d-1)!)}^\omega} \right)$,
        \item $\mathscr{C}_{\textup{\texttt{Diag}}} = \mathcal{O}\left(d\cdot \dfrac{n^{d\omega}}{{(d!)}^\omega} \right)$,
    \end{enumerate}
    where we recall that $d=\mathcal{O}(\sqrt{n})$, as in \Cref{thm:result}.
\end{theorem}

\begin{proof}[Proof of \Cref{thm:iterationCost} {(1.)}]
\textbf{Complexity of \texttt{Sgen}.}
We determine the overall complexity of \texttt{Sgen}, denoted by $\mathscr{C}_\texttt{Sgen}$, by first bounding the cost of a single call at a given degree $i$, and then summing over all $1 \le i \le d-1$.

Consider the diagonalized set $\mathcal{D}_{i+1}$. Each polynomial in $\mathcal{D}_{i+1}$ contains at most $\binom{n+i}{i}+1$ monomials. Due to the S-polynomial structure in \Cref{eq:spoly}, computing one such polynomial requires at most $\binom{n+i}{i}$ field additions. Since we generate $\binom{n+i-1}{i}$ distinct S-polynomials at iteration $i$, the number of operations performed is exactly $\binom{n+i}{i}\binom{n+i-1}{i}$.

Using the combinatorial identity $\binom{n+i}{i} = \frac{n+i}{n}\binom{n+i-1}{i}$, the total complexity over all $d-1$ iterations is:
\[
\mathscr{C}_\texttt{Sgen} = \sum_{i=1}^{d-1} \frac{n+i}{n} \binom{n+i-1}{i}^2.
\]
Since $d=\mathcal{O}(\sqrt{n})$, we have $i \le d-1 < n$, so the fractional multiplier $\frac{n+i}{n} = 1 + \frac{i}{n}$ is strictly bounded by $2$. Thus, it acts as an $\mathcal{O}(1)$ constant, simplifying the complexity sum to:
\[
\mathscr{C}_\texttt{Sgen} = \sum_{i=1}^{d-1} \mathcal{O}\left( \binom{n+i-1}{i}^2 \right) = \mathcal{O}\left( (d-1)\cdot \frac{n^{2(d-1)}}{{((d-1)!)}^2} \right),
\]
where the last equality follows directly from \Cref{lem:complexity_bound} with $h = d-1$ and exponent $c=2$.
\end{proof}

\begin{proof}[Proof of \Cref{thm:iterationCost} {(2.)}]
\textbf{Complexity of \texttt{Sred}.}
We determine the overall complexity of \texttt{Sred}, denoted by $\mathscr{C}_\texttt{Sred}$, by first bounding the cost of a single call at a given degree $i$, and then summing over all $1 \le i \le d-1$.

As established in \Cref{sec:onspolyred} and \Cref{alg:Sred}, the reduction of the generated S-polynomials is algebraically equivalent to computing the Schur complement $C - BA$. Since the top-left block of the joint system is the identity matrix, no matrix inversion is required. Consequently, the computational cost of this reduction step is strictly dominated by the rectangular matrix multiplication $B \times A$, as the subsequent matrix subtraction is asymptotically negligible. 

Let $B \in \mathbb{F}_q^{\ell \times k}$ and $A \in \mathbb{F}_q^{k \times r}$. By partitioning the matrices into square sub-blocks of dimension $u = \min(\ell, k, r)$, the complexity of this product using fast matrix multiplication is strictly bounded by $\mathcal{O}(\ell\cdot k\cdot r \cdot u^{\omega-3})$~\cite{DBLP:journals/jc/HuangP98}.

In our case, at iteration $i$, the block dimensions are given by:
\[
k = \binom{n+i}{i+1}, \quad \ell = \binom{n+i-1}{i}, \quad r = \binom{n+i}{i}.
\]
Recalling that $d=\mathcal{O}(\sqrt{n})$, in particular $d<n$, and the minimum among these dimensions is $u = \ell$. Substituting these variables into the complexity bound, the term $u^{\omega-3}$ multiplies with the existing $\ell$ to produce an $\ell^{\omega-2}$ factor. 

Thus, the total complexity of the S-polynomial reduction step across all iterations evaluates to:
\begin{align*}
    \mathscr{C}_\texttt{Sred} &= \sum_{i=1}^{d-1} \mathcal{O}\left(\binom{n+i-1}{i}^{\omega-2} \cdot \binom{n+i}{i} \cdot \binom{n+i}{i+1}\right)\\
    &= \sum_{i=1}^{d-1} \mathcal{O}\left(\frac{{(n+i)}^2}{n(i+1)}\cdot\binom{n+i-1}{i}^{\omega}\right).
\end{align*}
Since $d=\mathcal{O}(\sqrt{n})$, we have $i<d<n$, so we can bound the fractional multiplier, observing that $\frac{n+i}{n} < 2$, and $\frac{n+i}{i+1} \le n+i < 2n$. Thus, the entire fraction $\frac{{(n+i)}^2}{n(i+1)}$ is strictly bounded by $4n$,  i.e.\ it is $\mathcal{O}(n)$. Extracting this multiplicative factor, the complexity sum simplifies to:
\[
\mathscr{C}_\texttt{Sred} = \mathcal{O}\left(n\right) \cdot \sum_{i=1}^{d-1} \mathcal{O}\left(\binom{n+i-1}{i}^{\omega}\right) = \mathcal{O}\left((d-1)\cdot \frac{n^{1 + \omega(d-1)}}{{((d-1)!)}^\omega} \right),
\]
where the last equality follows directly from \Cref{lem:complexity_bound} with $h = d-1$ and exponent $c=\omega$.
\end{proof}

\begin{figure}[t]
    \centering
    \resizebox{0.8\linewidth}{!}{%
        \begin{tikzpicture}[
    every label/.append style={anchor=base, yshift=3pt},
    matparens/.style={
      inner sep=3pt,
      left delimiter={(},
      right delimiter={)}
    }
  ]
  \node[matparens] (block1) {
    \begin{tikzpicture}
      \node[lipicsGray, matrix={2}{2}, fill=lipicsBulletGray] (R) {};
      \draw[lipicsGray, line width=0.70mm] (R.north west) -- (R.south east) -- cycle;
      \node[font=\LARGE\bfseries, fill=white, draw=lipicsGray, line width=0.4pt, inner sep=2pt] at (R.center) {$i$};
      \node[lipicsBulletGray, matrix={2}{3}, fill=lipicsLightGray, right=2mm of R] (T) {};
      \node[black] at (T.center) {\LARGE$<i$};
    \end{tikzpicture}
  };
  \node[eqn, right=8mm of block1] (arrow) {$\implies$};
  \node[matparens, right=8mm of arrow] (block3) {
    \begin{tikzpicture}
      \node[lipicsGray, matrix={2}{2}, fill=white] (P) {};
      \draw[lipicsGray, line width=0.70mm] (P.north west) -- (P.south east) -- cycle;
      \node[font=\LARGE\bfseries, fill=white, draw=lipicsGray, line width=0.4pt, inner sep=2pt] at (P.center) [xshift=-2mm] {\LARGE$i$};
      \node[lipicsBulletGray, matrix={2}{3}, fill=lipicsLightGray, right=2mm of P] (L) {};
      \node[black] at (L.center) [xshift=-6mm] {\LARGE$<i$};
    \end{tikzpicture}
  };
\end{tikzpicture}%
    }
    \caption{\texttt{Diag}: Reduction to RREF, diagonal with high probability (\Cref{sec:probabilities2,sec:probabilitiesD}).}\label{fig:diagonalization}
\end{figure}
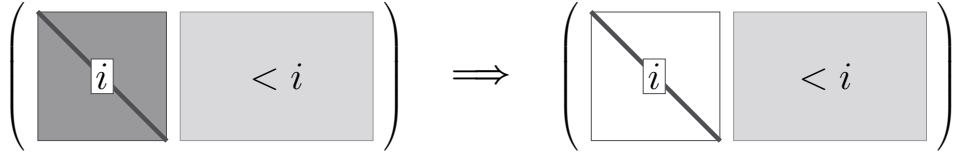\label{fig:step2a}

\begin{proof}[Proof of \Cref{thm:iterationCost} {(3.)}]
\textbf{Complexity of \texttt{Diag}.}
We determine the overall complexity of \texttt{Diag}, denoted by $\mathscr{C}_\texttt{Diag}$, by first bounding the cost of a single call at a given degree $i$, and then summing over all $1 \le i \le d$.

Let $1\le i \le d$ and consider the polynomial set $\mathcal{P}_{i}\subset \mathbb{F}_q[x_1,\dots,x_n]$. The corresponding coefficient matrix has $r = \binom{n+i-1}{i}$ rows and $c = \binom{n+i}{i}$ columns. By employing fast matrix multiplication techniques, a rectangular $r \times c$ matrix (with $r \le c$) can be reduced to RREF in $\mathcal{O}(r^{\omega-1} c)$ field operations~\cite{DBLP:journals/jal/IbarraMH82}, where $2<\omega\le3$ is the exponent of matrix multiplication.

Using the combinatorial identity $c = \frac{n+i}{n}\cdot r$, the computational cost for a single reduction step simplifies to:
\[
\mathcal{O}\left(r^{\omega-1} c\right) = \mathcal{O}\left( \frac{n+i}{n}\cdot r^\omega \right) = \mathcal{O}\left( \frac{n+i}{n}\cdot \binom{n+i-1}{i}^\omega \right).
\]

The total complexity of the RREF reduction is obtained by summing the contribution of the $d$ iterations over all degrees $1\leq i \leq d$:
\begin{align*}
    \mathscr{C}_\texttt{Diag}
    = \sum_{i=1}^{d}  \mathcal{O}\left( \frac{n+i}{n}\cdot\binom{n+i-1}{i}^\omega \right).
\end{align*}

Since $d=\mathcal{O}(\sqrt{n})$, we have $d<n$, so the fractional multiplier $\frac{n+i}{n}$ is strictly bounded by $2$. Thus, it acts as an $\mathcal{O}(1)$ constant, simplifying the complexity sum to:
\[
\mathscr{C}_\texttt{Diag} = \sum_{i=1}^{d} \mathcal{O}\left( \binom{n+i-1}{i}^\omega \right) = \mathcal{O}\left( d\cdot\frac{n^{d\omega}}{{(d!)}^\omega} \right),
\]
where the last equality follows directly from \Cref{lem:complexity_bound} with $h = d$ and  exponent $c=\omega$.
\end{proof}

\begin{proof}[Proof of \Cref{thm:result}]
    The total computational complexity of \Cref{alg:alg}, denoted by $\mathscr{C}_{tot}$, follows directly from summing the asymptotic costs of the three main procedures established in \Cref{thm:iterationCost}. Thus, we obtain:
    \begin{equation*}
        \mathscr{C}_{tot} = \mathscr{C}_{\texttt{Sgen}} + \mathscr{C}_{\texttt{Sred}} + \mathscr{C}_{\texttt{Diag}} = \mathcal{O}\left( d\cdot \dfrac{n^{d\omega}}{{(d!)}^\omega} \right).
    \end{equation*}
    The correctness of \Cref{alg:alg} is proven in \Cref{sec:correct}.
\end{proof}

\paragraph*{\belwe.}
This case follows directly from our general framework by setting $d = 2$. Consequently, \Cref{alg:alg} requires $n(n+1)/2$ samples and has a computational complexity of $\mathcal{O}\left({2\cdot n^{2\omega}}/{{(2!)}^\omega}\right) = \mathcal{O}\left(n^{2\omega}\right)$.

\paragraph*{\bslwe.}
By appending the constraint set $\{ x_j^2 - x_j = 0: 1 \le j \le n\}$, we restrict the system exclusively to square-free monomials. Under the condition that $n > 3d - 1$, this setting requires $\binom{n}{d}$ samples, yielding a computational complexity of $\mathcal{O}\left( d\cdot{n^{d\omega}}/{{(d!)}^\omega} \right)$. A more detailed discussion is given in~\Cref{sec:binarysecret}.

\section{Success Probability for \belwe}\label{sec:probabilities2}
We analyse the success probability of \Cref{alg:alg} separately for the two cases $d=2$ and $d>2$. Here we treat the case $d=2$, corresponding to the \belwe\ problem, where \Cref{alg:alg} performs a single iteration; the case $d>2$ is treated in \Cref{sec:probabilitiesD}. Throughout both sections the number of samples is $m=\binom{n+d-1}{d}$, and we will denote this quantity by $N$.

For $d=2$, \Cref{alg:alg} diagonalizes the homogeneous degree-$2$ part of the LWE polynomial system, which succeeds exactly when the coefficient matrix $\mathbf C\in\mathbb F_q^{N\times N}$ of that part has full rank, where each row of $\mathbf C$ is determined by the corresponding LWE sample vector $\mathbf a_i$ alone. The success probability is therefore
\[
p_{\mathrm{succ}}=\Pr\bigl[\operatorname{rank}\mathbf C=N\bigr],
\]
taken over the uniform samples $\mathbf a_1,\dots,\mathbf a_N\in\mathbb F_q^{n}$.

\paragraph*{Proof idea.}
If $\mathbf{C}$ does not have full rank, there is a nonzero vector $\boldsymbol\phi\in\mathbb F_q^{N}$ in its kernel. Since $\mathbf{C}$ is the homogeneous degree-$2$ part of the system, its entries depend only on the samples: the $i$-th row of $\mathbf{C}$ collects the values $a_{i,j}a_{i,k}$, for $j\le k$, of all degree-2 monomials, evaluated at $\mathbf a_i$. Indexing the entries of $\boldsymbol\phi$ by the same monomials, $\boldsymbol\phi$ is the coefficient vector of the quadratic form $Q_{\boldsymbol\phi} = \sum_{j\le k}\phi_{jk}x_j x_k$, and the $i$-th entry of $\mathbf{C}\boldsymbol\phi$ is exactly $Q_{\boldsymbol\phi}(\mathbf{a}_i)$. Thus $\mathbf{C}\boldsymbol\phi=\mathbf{0}$ means that $Q_{\boldsymbol\phi}$ vanishes at each $\mathbf a_i$. Hence, $p_{\mathrm{succ}}$ is at least $1$ minus the probability that some nonzero quadratic form $Q_{\boldsymbol\phi}$ vanishes at all the samples, which is the quantity we bound.

Given that the samples are uniform, $Q_{\boldsymbol\phi}$ vanishes at all of them with probability, say, $\rho_{\boldsymbol\phi}^{\,N}$, i.e. $\rho_{\boldsymbol\phi}$ is the fraction of $\F_q^n$ on which $Q_{\boldsymbol\phi}$ vanishes. It remains to sum this over all the forms that could occur, and the key point is that $\rho_{\boldsymbol\phi}$ depends only on their rank, so they can be grouped accordingly.

We therefore bound $\rho_{\boldsymbol\phi}$ in terms of the rank in \Cref{lemma:vanish}. Combined with the number of quadratic forms of each rank, this gives the general bound of \Cref{lemma:fullrank}. Estimating the resulting sum yields the explicit statement of \Cref{cor:easy}.

\paragraph*{Proof of success probability for $d=2$.}
We begin by computing the fraction of $\F_q^n$ on which a quadratic form $Q_{\boldsymbol\phi}$ vanishes, which depends only on its rank.

\begin{lemma}[Vanishing fraction of a quadratic form]\label{lemma:vanish}
Let $q$ be odd and let $Q\in\mathbb{F}_q[x_1,\dots,x_n]$ be a homogeneous quadratic form of rank $1\le r\le n$. Let $\rho_r$ denote the fraction of $\mathbb{F}_q^{n}$ on which $Q$ vanishes. Then
\[
\rho_r=\frac{1}{q}\quad(r\text{ odd}),\qquad \rho_r\le \frac{1}{q}+(q-1)\,q^{-r/2-1}\quad(r\text{ even}).
\]
\end{lemma}

\begin{proof}
Since $\operatorname{char}\mathbb{F}_q\neq2$, the form $Q$ can be diagonalized: there is an invertible linear change of variables (that preserves the number of zeros) after which $Q$ is equivalent to
\[
c_1x_1^2+\cdots+c_r{x_r}^2,\qquad c_1,\dots,c_r\in\mathbb{F}_q^{*},
\]
where $r=\operatorname{rank}Q$ is the number of nonzero coefficients. The remaining $n-r$ variables do not appear and are therefore free. Writing $Z(Q)$ for the number of zeros of $Q$ in $\mathbb{F}_q^{n}$ and $Z_r$ for the number of zeros of $c_1x_1^2+\cdots+c_r{x_r}^2$ in $\mathbb{F}_q^{r}$, the free variables contribute a factor $q^{n-r}$:
\[
Z(Q)=q^{\,n-r}\,Z_r,\qquad\text{so}\qquad \rho_r=\frac{Z(Q)}{q^{n}}=\frac{Z_r}{q^{r}}.
\]

The form $c_1x_1^2+\cdots+c_r{x_r}^2$ is nondegenerate in $\mathbb F_q[x_1,\dots,x_r]$, so by~\cite[Thms.~6.26--6.27]{Lidl_Niederreiter_1996}, with discriminant $\Delta=c_1\cdots c_r$ and quadratic character $\eta$ ($|\eta|\le1$),
\[
Z_r=q^{\,r-1}\ (r\text{ odd}),\qquad Z_r=q^{\,r-1}+(q-1)\,q^{\,r/2-1}\,\eta\left({(-1)}^{r/2}\Delta\right)\ (r\text{ even}).
\]
For even $r$, since $|\eta|\le1$, this gives $Z_r\le q^{\,r-1}+(q-1)\,q^{\,r/2-1}$. Dividing by $q^{r}$,
\[
\rho_r=\frac{1}{q}\ (r\text{ odd}),\qquad \rho_r\le\frac{1}{q}+(q-1)\,q^{-r/2-1}\ (r\text{ even}).
\]
\end{proof}

Let $M(n,r)$ denote the number of symmetric $n\times n$ matrices of rank $r$ over $\mathbb F_q$, given by~\cite[Thm.~2]{MacWilliams1969}:
\begin{equation}\label{eq:sym_matrices}
  M(n,r)=\prod_{i=1}^{\lfloor r/2\rfloor}\frac{q^{2i}}{q^{2i}-1}\prod_{i=0}^{r-1}\bigl(q^{\,n-i}-1\bigr).
\end{equation}

\begin{lemma}[Full rank of $\mathbf C$]\label{lemma:fullrank}
Let $q$ be odd, $N=\binom{n+1}{2}$, and let $\mathbf C\in\mathbb F_q^{N\times N}$ have as rows the coefficient vectors of the homogeneous degree-$2$ polynomials defined from $N$ uniform samples $\mathbf a_1,\dots,\mathbf a_N\in\mathbb F_q^{n}$. Then
\[
\Pr[\operatorname{rank}\mathbf C=N]\ \ge\ 1-\frac{1}{q-1}\sum_{r=1}^{n} M(n,r)\,\rho_r^{\,N}.
\]
\end{lemma}

\begin{proof}
Assume that $\operatorname{rank}\mathbf{C}<N$. Then the rows of $\mathbf{C}$ are dependent $\iff$ $\mathbf{C}\boldsymbol\phi=\mathbf{0}$ for some $\mathbf{0}\neq\boldsymbol\phi\in\mathbb{F}_q^N$. 
Thanks to the structure of $\mathbf{C}$, the vector $\boldsymbol\phi$ is the coefficient vector of a quadratic form $Q_{\boldsymbol\phi}$ in $n$ variables, and $\mathbf C\boldsymbol\phi=\mathbf 0$ means precisely that $Q_{\boldsymbol\phi}(\mathbf a_i)=0$ for every $i\in\{1,\dots,N\}$. 

For a fixed $\boldsymbol\phi$, write $\rho_{\boldsymbol\phi}:=\Pr_{\mathbf a}[Q_{\boldsymbol\phi}(\mathbf a)=0]$. Since the samples are uniform and i.i.d., we have
\[
\Pr\bigl[Q_{\boldsymbol\phi}(\mathbf a_i)=0,\ \forall i\in\{1,\dots,N\}\bigr]=\prod_{i=1}^{N}\Pr[Q_{\boldsymbol\phi}(\mathbf a_i)=0]=\rho_{\boldsymbol\phi}^{\,N}.
\]

A nonzero $\boldsymbol\phi$ and its $q-1$ scalar multiples $c\boldsymbol\phi$ define the same ``vanishing event'' ($Q_{c\boldsymbol\phi}=c\,Q_{\boldsymbol\phi}$ has the same zeros), so the distinct events correspond to the $\tfrac{q^N-1}{q-1}$ one-dimensional directions. The union bound over these distinct events gives
\[
\Pr[\operatorname{rank}\mathbf C<N]\ \le\ \frac{1}{q-1}\sum_{\boldsymbol\phi\neq0}\rho_{\boldsymbol\phi}^{\,N}.
\]

We now count by rank. A quadratic form of rank $r$ corresponds to a symmetric $n\times n$ matrix of rank $r$, so the number of
rank-$r$ forms equals $M(n,r)$ from~\Cref{eq:sym_matrices}.

By \Cref{lemma:vanish} each rank-$r$ form vanishes on a fraction at most $\rho_r$ of $\mathbb F_q^{n}$ (with equality for odd $r$).
Grouping the sum by rank,
\[
\sum_{\boldsymbol\phi\neq0}\rho_{\boldsymbol\phi}^{\,N} =\sum_{r=1}^{n}\ \sum_{\operatorname{rank}Q_{\boldsymbol\phi}=r}\rho_{\boldsymbol\phi}^{\,N} \ \le\ \sum_{r=1}^{n} M(n,r)\,\rho_r^{\,N},
\]
hence $\Pr[\operatorname{rank}\mathbf C<N]\le\frac{1}{q-1}\sum_{r=1}^{n} M(n,r)\,\rho_r^{\,N}$. Subtracting from $1$ gives the claim.
\end{proof}

\begin{corollary}[Success probability of \Cref{alg:alg} for $d=2$]\label{cor:easy}
Let $q$ be odd and $N=\binom{n+1}{2}$. If $n\ge 13$, then
\[
\Pr[\operatorname{rank}\mathbf C=N]\ \ge\ 1-\frac{2}{q-1}.
\]
\end{corollary}

\begin{proof}
By \Cref{lemma:fullrank} it suffices to prove $S:=\sum_{r=1}^{n}M(n,r)\,\rho_r^{N}<2$.

For even $r$, we can write 
\begin{gather*}
    \rho_r \le \frac{1}{q}+(q-1)\,q^{-r/2-1}
    = \frac{1}{q}\left(1+(q-1)\,q^{-r/2}\right)\\
    \implies\quad \rho_r^N \le \frac{1}{q^N}{\left(1+(q-1)\,q^{-r/2}\right)}^N.
\end{gather*}
Splitting $S$ into odd and even ranks, and recalling \Cref{lemma:vanish}, we have that
\begin{align*}
  S &= \sum_{r\text{ odd}}M(n,r)\,\rho_r^{N} + \sum_{r\text{ even}}M(n,r)\,\rho_r^{N} \\
  &\le \frac{1}{q^{N}}\sum_{r=1}^{n}M(n,r) + \frac{1}{q^{N}}\sum_{r\ \mathrm{even}}M(n,r)\Bigl[{\left(1+(q-1)\,q^{-r/2}\right)}^N-1\Bigr].
\end{align*}

The number of symmetric $n\times n$ matrices over $\mathbb{F}_q$ is $q^{\binom{n+1}{2}}=q^{N}$, hence $\sum_{r=1}^{n}M(n,r)=q^{N}-1$ and the first term equals $\frac{q^{N}-1}{q^{N}}<1$. It remains to show $E<1$, where $E$ denotes the second term.

Start by bounding $M(n,r)$. For the first product term, the Weierstrass product inequality $\left(\prod_{i=1}^{n}(1-a_i)\ge 1-\sum_{i=1}^{n} a_i,\, \text{for }a_i\in[0,1]\right)$ and the geometric series $\left(\sum_{i=1}^{n}a^{i}=\frac{a(1-a^n)}{1-a}\right)$ give
\[
\prod_{i=1}^{\lfloor r/2\rfloor}\frac{q^{2i}}{q^{2i}-1}
=\prod_{i=1}^{\lfloor r/2\rfloor}\frac{1}{1-q^{-2i}}
\le \frac{1}{1-\sum_{i=1}^{\lfloor r/2\rfloor}q^{-2i}}
= \frac{1}{1-\frac{1-q^{-r}}{q^2-1}}
<\frac{q^2-1}{q^2-2}.
\]
For the second product, the finite arithmetic series $\left(\sum_{i=0}^{r-1}(n-i)=rn-\binom r2\right)$ gives,
\[
\prod_{i=0}^{r-1}\bigl(q^{n-i}-1\bigr)<\prod_{i=0}^{r-1}q^{n-i}=q^{\,rn-\binom r2}.
\]
Using the identity $rn-\binom r2-N=-\binom{n-r+1}{2}$, we have that $\frac{1}{q^{N}}q^{\,rn-\binom r2}=q^{-\binom{n-r+1}{2}}$, and
\[
E<\frac{q^2-1}{q^2-2}\sum_{r\ \mathrm{even}}q^{-\binom{n-r+1}{2}}
\Bigl[{\bigl(1+(q-1)\,q^{-r/2}\bigr)}^{N}-1\Bigr].
\]

To bound the bracketed term, we use the algebraic identity ${(1+x)}^N-1=x\sum_{j=0}^{N-1}{(1+x)}^j$. Since $x_r:= (q-1)q^{-r/2} > 0$, we can establish an upper bound by taking $N$ times the largest term in the sum:
\[
{(1+x_r)}^{N}-1 \ \le\ x_r \cdot N{(1+x_r)}^{N-1} \ \le\ Nx_r\,{(1+x_r)}^{N}.
\]
For the remaining power, the inequality $1+x\le e^x$ gives ${(1+x_r)}^N\le e^{Nx_r}$, and the identity $q^{m}=e^{m\ln q}$ lets us write $e^{Nx_r}=q^{Nx_r/\ln q}$. We bound this by $q^{y}$ for a suitable exponent $y$. Since the term $q^y$ will be multiplied by $q^{-\binom{n-r+1}{2}}$, this yields a final contribution $q^{\,y-\binom{n-r+1}{2}}$. To make the exponent linear in $r$, we use $\binom{n-r+1}{2}=\binom{n-r}{2}+(n-r)$ and take $y=\binom{n-r}{2}$, leaving the linear residual $-(n-r)$. This choice vanishes at the largest even rank, where $\binom{n-r}{2}=0$ would force $q^{Nx_r/\ln q}\le q^{0}=1$ (impossible, since $Nx_r>0$). We therefore take $y=\binom{n-r}{2}+1$.
It remains to verify that this choice satisfies $\tfrac{Nx_r}{\ln q}\le\binom{n-r}{2}+1$, which rewrites as $\binom{n+1}{2}\le\frac{q^{r/2}\ln q}{q-1}\bigl(\binom{n-r}{2}+1\bigr)$, whose right-hand side increases in $q$ and is smallest over even $r$ at $r=2$. Hence the constraint is tightest at $q=3$ and $r=2$, where it reads $\frac{2}{3\ln 3}\,n(n+1)\le(n-2)(n-3)+2$, valid for $n\ge13$.
Therefore
\[
{(1+x_r)}^N-1\ \le\ Nx_r\,q^{\binom{n-r}{2}+1},
\]
and, using $\binom{n-r+1}{2}-\binom{n-r}{2}=n-r$,
\[
q^{-\binom{n-r+1}{2}}\bigl[{(1+x_r)}^{N}-1\bigr]\ \le\ Nx_r\,q^{\,1-(n-r)}\ =\ N(q-1)\,q^{-n/2}\,q^{\,1-(n-r)/2}.
\]
Summing over the even ranks $2\le r\le n$,
\[
E\ <\ \frac{q^2-1}{q^2-2}\,N(q-1)\,q^{-n/2}\sum_{r\ \mathrm{even}}q^{\,1-(n-r)/2},
\]
the geometric sum gives
\[
\sum_{r\ \mathrm{even}}q^{\,1-(n-r)/2}
=q^{\,1-n/2}\sum_{r\ \mathrm{even}}q^{\,r/2}
=q^{\,1-n/2}\,q\,\frac{q^{\lfloor n/2\rfloor}-1}{q-1}
<\frac{q^{2}}{q-1},
\]
where the last step uses $q^{\lfloor n/2\rfloor}-1<q^{\,n/2}$ (valid for both parities of $n$).
Hence
\[
E\ <\ \frac{q^2-1}{q^2-2}\,N(q-1)\,q^{-n/2}\cdot\frac{q^{2}}{q-1}
= \frac{q^2-1}{q^2-2}\, N\,q^{\,2-n/2}.
\]
For a fixed $n$ this bound decreases in $q$ as soon as the exponent $2-\tfrac n2$ is negative, i.e.\ for $n\ge 5$. Together with $\frac{q^2-1}{q^2-2}$ decreasing in $q$, a product of positive decreasing functions is decreasing, so the expression decreases in $q$. The worst case among odd $q\ge3$ is therefore $q=3$. Setting $q=3$ leaves the function of $n$, for $n\ge5$:
\[
\frac{q^2-1}{q^2-2}\, N\,q^{\,2-n/2}\bigg|_{q=3}
=\frac{8}{7}\,3^{\,2-n/2}\,\binom{n+1}{2}=:F(n).
\]
Consecutive values satisfy
\[
\frac{F(n+1)}{F(n)}=\frac{n+2}{n}\cdot\frac{1}{\sqrt3}<1\qquad \text{for}\qquad n\ge3,
\]
so $F$ is decreasing. Since $F(12)\approx1.10>1>F(13)\approx0.74$, it first drops below
$1$ at $n=13$ and stays below for all $n\ge13$.
Hence the bound is $<1$ for all $n\ge13$ and odd $q\ge3$, so $E<1$ and
$S\le\frac{q^{N}-1}{q^{N}}+E<2$, proving the claim.
\end{proof}

\begin{example}
The bound of \Cref{cor:easy} does not depend on the dimension $n$ (for $n\ge13$), and the success probability $p_{\mathrm{succ}}$ approaches $1$ as $q$ grows. Even for small moduli \Cref{alg:alg} already succeeds with high probability:
\begin{table}[ht]
\centering
\setlength{\tabcolsep}{10pt}
\renewcommand{\arraystretch}{0.8}
\begin{tabular}{rrrrcc}
\toprule
$n$ & $d$ & $q$ & $N$ & Experimental & Bound (Cor.~\ref{cor:easy})  \\
\midrule
13 & 2 & 17    & 91  & $0.9442$ & $0.8750$ \\
13 & 2 & 31    & 91  & $0.9641$ & $0.9333$ \\
13 & 2 & 97    & 91  & $0.9871$ & $0.9792$ \\
13 & 2 & 251   & 91  & $0.9961$ & $0.9920$ \\
13 & 2 & 1021  & 91  & $0.9989$ & $0.9980$ \\
13 & 2 & 3329  & 91  & $0.9998$ & $0.9994$ \\
\midrule
20 & 2 & 17    & 210 & $0.9395$ & $0.8750$ \\
20 & 2 & 31    & 210 & $0.9691$ & $0.9333$ \\
20 & 2 & 97    & 210 & $0.9894$ & $0.9792$ \\
20 & 2 & 251   & 210 & $0.9956$ & $0.9920$ \\
20 & 2 & 1021  & 210 & $0.9993$ & $0.9980$ \\
20 & 2 & 3329  & 210 & $0.9996$ & $0.9994$ \\
\midrule
40 & 2 & 17    & 820 & $0.9427$ & $0.8750$ \\
40 & 2 & 31    & 820 & $0.9673$ & $0.9333$ \\
40 & 2 & 97    & 820 & $0.9900$ & $0.9792$ \\
40 & 2 & 251   & 820 & $0.9963$ & $0.9920$ \\
40 & 2 & 1021  & 820 & $0.9994$ & $0.9980$ \\
40 & 2 & 3329  & 820 & $0.9997$ & $0.9994$ \\
\bottomrule
\end{tabular}
\caption{Degree $d=2$, over $10^{4}$ trials per parameter set: observed success probability against the lower bound $1-2/(q-1)$ of \Cref{cor:easy}. Here $N=\binom{n+1}{2}$ is the number of samples drawn in each trial.}\label{tab:experiments-d2}
\end{table}

\end{example}

\section{\texorpdfstring{Success Probability for LWE ($d>2$)}{Success Probability for LWE (d>2)}}\label{sec:probabilitiesD}
\Cref{alg:alg} performs $d$ iterations, and succeeds only if all of them do. We analyse the first one separately from the subsequent ones, since the initial system exhibits a precise structure inherited from the definition of LWE which is lost afterwards. The two estimates are then combined into the overall success probability in \Cref{cor:easy2}.

\subsection{First iteration}
At this stage \Cref{alg:alg} has to diagonalize the degree-$d$ part of the system, and this succeeds exactly when the corresponding coefficient matrix has full rank, that is, when its rows are linearly independent. We therefore bound the probability of a successful diagonalization by studying the probability that these rows (seen as images of the samples under the degree-$d$ Veronese map) are linearly independent.
Recall that the degree-$d$ Veronese map sends a point to the vector of all degree-$d$ monomials evaluated at it; it is a standard construction in algebraic geometry~\cite{MR1182558}.

Let $a_1, \dots, a_N \in \F_q^n$, with $N = \binom{n+d-1}{d}$, be the uniformly random LWE sample vectors over $\F_q$, and let $2< d \le q-1$. Let
\[
    \nu_d : \F_q^n \longrightarrow \F_q^N, \qquad
    \mathbf x \longmapsto {\left(\mathbf x^{\alpha}\right)}_{|\alpha| = d},
\]
be the degree-$d$ Veronese map, the $N$ coordinates being indexed by the degree-$d$ monomials in some fixed order.

If $q > Nd+2$, i.e.\ large fields, the probability that the $N$ vectors $\nu_d(a_1),\allowbreak \dots,\allowbreak \nu_d(a_N)$ are linearly independent was already given in~\cite[Lem.~5]{EC:NMSU25} as  \[
    \Pr\left[v_1,\dots,v_N\text{ lin.\ ind.}\right]\;\ge\;
    1-\frac{Nd}{q}.
\]
However, for many practical LWE instances, the condition on the field size $q$ is not met~\cite{NISTPQC:CRYSTALS-KYBER22,FRODOKEM}, and therefore the above bound is not applicable. In this section we provide a new bound holding for any $q \ge 3d$.

We want to compute the probability that the vectors $v_\ell := \nu_d(a_\ell)$ are linearly independent, for $\ell=1,\dots,N$. We compute it iteratively: assuming $v_1, \dots, v_\ell$ are already independent, we bound the probability that $v_{\ell+1}$ is \emph{dependent} on them, i.e.\ that $v_{\ell+1} \in V_\ell:= \langle v_1, \dots, v_\ell \rangle$, $\dim V_\ell = \ell$. 
The independence probability then follows from the complementary events.
The quantity we want to bound is
\[
p_\ell := \Pr_{a_{\ell+1}}\!\bigl[\, \nu_d(a_{\ell+1}) \in V_\ell \,\bigr]
= \frac{\left|\{\, b \in \F_q^n : \nu_d(b) \in V_\ell \,\}\right|}{q^n}
= \frac{\left|\nu_d^{-1}(V_\ell)\right|}{q^n},
\]
which is a function of the configuration $a_1,\dots,a_\ell$, hence a random variable.

Since $v_1, \dots, v_\ell$ form a basis of $V_\ell$, a point $b \in \nu_d^{-1}(V_\ell)$ is characterised by the existence of a (unique) coefficient vector $\lambda = (\lambda_1, \dots, \lambda_\ell) \in \F_q^\ell$ with
\begin{equation}\label{linear_dependence}
    \nu_d(b) \;=\; \sum_{j=1}^{\ell} \lambda_j\, v_j.
\end{equation}

In what follows $\mathbb{P}^{n-1}$ denotes the projective space of dimension $n-1$, that is, the set of lines through the origin of $\F_q^{n}$. We write $\mathbb{P}^{n-1}_{\overline{\F}_q}$ when the same construction is taken over the algebraic closure.

\paragraph*{Proof idea.}
We show that
\[
    \Pr\left[v_1,\dots,v_N\text{ lin.\ ind.}\right]\;\ge\;
    1-\frac{d}{q}-3{\left(\frac{d}{q}\right)}^{2},
\]
proceeding one sample at a time: at each step we bound the probability $p_\ell$ that a fresh
uniform $b\in\F_q^n$ has $\nu_d(b)$ in the span $V_\ell$ of the previous $\ell$ images.

The starting observation is that this condition admits a geometric reformulation. Writing the dependence $\nu_d(b)=\sum_j\lambda_j v_j$ coordinate-wise and eliminating the scalars $\lambda_j$ leaves a system of $N-\ell$ homogeneous polynomials of degree $d$, namely those vanishing at $a_1,\dots,a_\ell$; and $\nu_d(b)\in V_\ell$ precisely when $b$ is a common zero of all of them. Thus, $p_\ell$ is the density in $\F_q^n$ of the common zero locus of a linear system of degree-$d$ polynomials, and the problem becomes that of counting $\F_q$-points on that locus. The key point is that this count depends only on the degree and the dimension of the locus. The B\'ezout inequality bounds the degree unconditionally, while the dimension is as specified in Assumption~\ref{ass:generic}. We recall that B\'ezout inequality says that if $V_1,\dots,V_r$ are subvarieties of $\mathbb{P}^{n-1}$ and $Z_1,\dots,Z_t$ are the irreducible components of $\cap_i V_i$, then $\sum_j \deg(Z_j) \le \prod_i \deg(V_i)$~\cite[Ex.~8.4.6]{Fulton1984}; in particular, $r$ hypersurfaces of degree $d$ cut out a locus of degree at most $d^{\,r}$.

We therefore bound $p_\ell$ in terms of the number $r_\ell$ of equations actually exploited in \Cref{prop:per-step}. Chaining these bounds over the $N$ steps yields the estimate of \Cref{thm:main}. Since $r_\ell$ decreases as $\ell$ grows, the last step dominates the whole sum, and there a single polynomial remains and no assumption is needed.

\paragraph*{Proof of success probability for $d>2$.}
We start from the trivial solutions of~\Cref{linear_dependence}, which give a lower bound on $p_\ell$ and serve as a consistency check on the upper bounds that follow. We stress that what we obtain is a lower bound on the probability, not an exact value: as $q$ grows it approaches the values observed experimentally.

\paragraph*{Lower bound.}
The $\ell$ original points defining $V_\ell$ solve~\Cref{linear_dependence}: taking $b = a_j$ gives $\nu_d(a_j) = v_j$, i.e.\ $\lambda = e_j$. Moreover $\nu_d$ is homogeneous of degree $d$, so $\nu_d(\zeta b) = \zeta^d \nu_d(b)$ for every $\zeta \in \F_q$.
Since $V_\ell$ is a subspace, every scalar multiple $\zeta a_j$ is again a solution (with $\lambda = \zeta^d e_j$).
Furthermore, the $a_j$ are nonzero and pairwise non-parallel: if $a_k = \zeta a_j$ with $k \ne j$, or $a_j = 0$, then $v_k = \zeta^d v_j$ or $v_j = 0$, contradicting $\dim V_\ell = \ell$. Hence each line defined by $a_j$ meets the others only at the origin, and
\begin{equation}\label{eq:lower_bound}
    \left|\nu_d^{-1}(V_\ell)\right| \ge 1 + \ell(q-1)
    \qquad\Longrightarrow\qquad
    p_\ell \ge \frac{1 + \ell(q-1)}{q^{n}} .
\end{equation}

\begin{remark}\label{rem:exact-small-i}
Any $k \le d+1$ pairwise non-parallel nonzero vectors $b_1,\dots,b_k$ have linearly independent Veronese images: they define $k$ distinct Veronesean points of degree $d$, and any $d+1$ of these are independent over an arbitrary field by~\cite[Thm.~1.1]{kantor2011}. 
\end{remark}

\paragraph*{On the number of solutions of~\eqref{linear_dependence}.}
Since $\nu_d(\mathbf{x})$ is the vector of all degree-$d$ monomials $m_1(\mathbf{x}), \dots, m_N(\mathbf{x})$ in $\mathbf{x}$,~\Cref{linear_dependence} can be rewritten as the system of $N$ equations $m_k(\mathbf{x}) \;=\; \sum_{j=1}^{\ell} \lambda_j\, m_k(a_j), \ k = 1, \dots, N$, that is given by
\[
\begin{aligned}
    \sum_{j=1}^{\ell} \lambda_j\, {(a_j)}_1^{d} &=  x_1^{d}\\
    \sum_{j=1}^{\ell} \lambda_j\, {(a_j)}_1^{d-1}{(a_j)}_2 &=  x_1^{d-1}x_2 \\
    &\ \vdots \\
    \sum_{j=1}^{\ell} \lambda_j\, {(a_j)}_n^{d} &=  x_n^{d}
\end{aligned}
\]
which is a system of $N$ equations and $n+\ell$ unknowns, of degree $d$ in $\mathbf{x}$ and linear in $\lambda$. For a fixed value of $\mathbf{x}$, since the $v_j$ are independent by hypothesis, the system is overdetermined and has \emph{at most one} solution $\lambda$.

Since counting the number of such systems having exactly one solution is equivalent to computing the exact value of $\left|\nu_d^{-1}(V_\ell)\right|$, our aim is now to give an upper bound on the number of solutions.

\paragraph*{Through B\'ezout Bound.}
The polynomial system just defined can be represented by its coefficient matrix $M \in \F_q^{N \times (\ell+N)}$. We partition the matrix as $M = [M_1 \mid M_2]$, where $M_1 \in \F_q^{N \times \ell}$ is the matrix of the constants associated with the scalar unknowns $\lambda_1, \dots, \lambda_\ell$, and $M_2 \in \F_q^{N \times N}$ is the matrix associated with the degree-$d$ monomials $\nu_d(x_1, \dots, x_n)$, that is, $M_2 = I_N$.
 
By applying Gaussian elimination (diagonalization) to the rows of $M$ to put the first $\ell$ columns into reduced row echelon form, we obtain a new matrix $M' = [M'_1 \mid M'_2]$. Because the $\ell$ previously sampled vectors are linearly independent, $M_1$ has rank $\ell$, so that, after a suitable permutation of the rows, $M'_1$ contains the identity matrix in its first $\ell$ rows and zeros in its bottom $N-\ell$ rows.

Consequently, the last $N-\ell$ rows of $M'$ represent a new polynomial system where the $\lambda_j$ variables have been completely eliminated. Because the monomials $\nu_d(x)$ are strictly of degree $d$, this yields a system of $N-\ell$ homogeneous equations of degree $d$ in the $n$ variables $x_1, \dots, x_n$. Its solutions are exactly the points of $\nu_d^{-1}(V_\ell)$, and the polynomials occurring in it span the space
\[
    W_\ell \;=\; \left\{\, f \text{ of degree } d \;:\; f(a_1) = \cdots = f(a_\ell) = 0 \,\right\},
\]
of dimension $N-\ell$. We write $\overline W_\ell$ for the corresponding space of polynomials with coefficients in $\overline\F_q$, of the same dimension, and denote by
\[
    B_\ell \;:=\; Z\left(\overline W_\ell\right) \;\subseteq\; \mathbb{P}^{n-1}_{\overline{\F}_q}
\]
the associated projective base locus. 

Since the elements of $W_\ell$ are homogeneous, the solution set $\nu_d^{-1}(V_\ell) \subseteq \F_q^n$ is a cone: it contains the origin, and with any nonzero $b$ it contains the whole punctured line $\F_q^{*}b$, of $q-1$ points. Its nonzero points are therefore partitioned into punctured lines, one for each point of $B_\ell(\F_q)$, whence
\[
    \left|\nu_d^{-1}(V_\ell)\right| \;=\; (q-1)\,\left|B_\ell(\F_q)\right| + 1
\]
and, dividing by $q^n$,
\begin{equation}\label{eq:cone}
    p_\ell \;=\; \frac{(q-1)\,\left|B_\ell(\F_q)\right| + 1}{q^n} .
\end{equation}
Since $B_\ell$ is cut out by $N-\ell$ forms in $\mathbb{P}^{n-1}$, each of which drops the dimension by one, its expected dimension is $(n-1)-(N-\ell)$; and since $B_\ell$ always contains the $\ell$ sampled points (see~\Cref{eq:lower_bound}) it is never empty, whence the expected dimension $\max\big(0,\, n-1-(N-\ell)\big)$. We formalize the expected behaviour as follows.

\begin{assumption}\label{ass:generic}
For every $1 \le \ell \le N-2$ there exist
\[
    r_\ell \;:=\; \min\left(n-1,\ N-\ell\right)
\]
forms in $\overline W_\ell$ whose common projective zero locus in $\mathbb{P}^{n-1}_{\overline{\F}_q}$ has the expected codimension $r_\ell$ (equivalently, forming a regular sequence, see \Cref{sec:dreg}).
The two bounds defining $r_\ell$ have different origins: $n-1$ is the largest codimension available in $\mathbb{P}^{n-1}$, beyond which the zero locus is already $0$-dimensional and further polynomials bring nothing; $N-\ell$ is $\dim W_\ell$, i.e.\ the number of equations at our disposal. 
\end{assumption}

\paragraph*{Discussion on Assumption~\ref{ass:generic}.}
Assumption~\ref{ass:generic} is used to count the points of $B_\ell(\F_q)$, and hence to bound $\Pr\left[v_1,\dots,v_N\text{ lin.\ ind.}\right]$. The reason regularity is needed is simple. We want each added equation to cut the solution set dimension down by one, so that after
$r_\ell$ equations the set is as small as expected and B\'ezout's theorem bounds the number of points. This is precisely what a regular sequence guarantees: the Krull dimension drops by exactly one at each step, resulting in $n -1 - r_\ell$ after adding $r_\ell$ equations~\cite[Thm.~2.1.2]{Bruns_Herzog_1998}. If regularity fails, the equations may be partly redundant, the solution set can be larger than expected, and B\'ezout bound cannot be applied. We emphasize that we do not assume all of $\overline W_\ell$ to be regular: we assume only that $r_\ell = \min(n-1, N-\ell)$ of its homogenous polynomials form a regular sequence, which is a weaker requirement, since $\overline W_\ell$ has dimension $N-\ell$ and, except when $N-\ell < n-1$, we are free to choose which polynomials to use.

Note that~\cite{EC:NMSU25} also makes use of regularity, but there are important differences between our assumption and that of~\cite{EC:NMSU25}.
The first difference stands in the application target:~\cite{EC:NMSU25} needs the semi-regularity condition on the whole LWE polynomial system in order to bound the time complexity of the algorithm. \textit{We, however, apply the regularity assumption (not on the LWE polynomial system but) to tackle the problem of linear independence of Veronese's embeddings, which in turn is needed to bound the success probability of our algorithm}.
We stress that no Hilbert series or degree of regularity~\cite{BardetFaugereSalvyYang2005,BardetFaugereSalvy2004} appears in our analysis.

It is also worth noting that Assumption~\ref{ass:generic} helps us to improve the probability bound compared to~\cite{EC:NMSU25}. The invertibility of the matrix representing the Veronese's embeddings can be proved without any assumption, leading to a weak probability bound as in~\cite[Lem.~5]{EC:NMSU25} (using  the Schwartz-Zippel lemma~\cite{DBLP:journals/jacm/Schwartz80,DBLP:conf/eurosam/Zippel79}).
This bound of~\cite{EC:NMSU25} is vacuous for practical parameter choice. Our assumption, supported by experiments, gives us a stronger probability bound.

Assumption~\ref{ass:generic} does not hold with probability $1$, as it can fail on degenerate configurations, where the simplest is collinearity, e.g. if a line in $\mathbb{P}^{n-1}$ contains $d+1$ of the samples.
However, our experiments indicate that it holds with high probability. Hence, \textit{an interesting open question} is to formally prove that Assumption~\ref{ass:generic} holds with high probability. 

\begin{proposition}\label{prop:per-step}
Let $2< d\le q-1$ and set $r_\ell := \min(n-1,\,N-\ell)$. Under
Assumption~\ref{ass:generic}, for $1 \le \ell \le N-2$,
\[
    p_\ell \;\le\; {\left(\frac{d}{q}\right)}^{\!r_\ell} .
\]
For $\ell = N-1$, where $r_\ell = 1$, the bound $p_{N-1} \le d/q$ holds unconditionally; so does $p_\ell \le d/q$ for every $1 \le \ell \le N-1$.
\end{proposition}

\begin{proof}
We bound the solutions of a \emph{selected subsystem}; considering the remaining equations would only further reduce the variety of the ideal, so the resulting bound applies to the full system. Throughout we use~\Cref{eq:cone}. 

We distinguish three cases satisfying $r_\ell = \min(n-1,\,N-\ell)$. In Case~1 there are more equations than the ambient dimension allows to exploit, $r_\ell = n-1$, and the selected subsystem has finitely many solutions. In Cases~2 and~3 all the available equations are used, $r_\ell = N-\ell$, and the zero locus is positive-dimensional, so that its $\F_q$-rational points must be counted; Case~3 is the extreme situation $\ell = N-1$ of a single equation, where Assumption~\ref{ass:generic} plays no role.

\textbf{Case 1:} $1 \le \ell \le N-n+1$. Here at least $n-1$ equations are available, so $r_\ell = n-1$. By Assumption~\ref{ass:generic} we select $n-1$ polynomials whose common projective zero locus $X \subseteq \mathbb{P}^{n-1}$ has codimension $n-1$, i.e.\ dimension zero. By the B\'ezout inequality (see, e.g.,~\cite[\S2.1]{Cesaratto_2013} or~\cite{Fulton1984}), which requires no complete intersection hypothesis, $X$ consists of at most $d^{\,n-1}$ points. Since each projective solution corresponds to a line through the origin in affine space (yielding $q-1$ non-zero points), and adding the origin itself, the number of affine solutions is bounded by $d^{\,n-1}(q-1)+1 \le d^{\,n-1}q$, whence $p_\ell \le {(d/q)}^{n-1}$.
 
\textbf{Case 2:} $N-n+1 < \ell \le N-2$. Let $j = N-\ell$, so $2 \le j \le n-2$ and $r_\ell = j$. We select all $j$ polynomials, i.e.\ a basis of $\overline W_\ell$, so that $X = B_\ell$; by Assumption~\ref{ass:generic} it has codimension $j$, hence dimension $n-1-j > 0$, and by the B\'ezout inequality the degrees of its irreducible components sum to at most $d^{\,j}$. Since $X$ is not finite, it is the number of its $\F_q$-\emph{rational} points that we must bound: a projective variety of dimension $t$ and degree $\delta$ has at most $\delta\,(q^{t}+\cdots+q+1)$ rational points~\cite[\S2.2]{Cesaratto_2013}. Applying this to each irreducible component $Z_t$ of $X$ and using $\dim Z_t \le n-1-j$,
\[
    \left|X(\F_q)\right| \;\le\; \sum_t \deg Z_t \cdot \frac{q^{\,n-j}-1}{q-1}
    \;\le\; d^{\,j}\,\frac{q^{\,n-j}-1}{q-1},
\]
so that, passing to the affine cone, the number of affine solutions is at most $d^{\,j}(q^{\,n-j}-1)+1 \le d^{\,j}q^{\,n-j}$ and $p_\ell \le {(d/q)}^{j}$.

\textbf{Case 3:} $\ell = N-1$, i.e.\ $j = 1$. We are left with exactly one equation in $n$ variables, and here Assumption~\ref{ass:generic} is vacuous: $\overline W_{N-1}$ is spanned by a single nonzero form $f$ of degree $d$, whose zero locus is a hypersurface, of dimension exactly $n-2$ and degree at most $d$ unconditionally. The bound of Case~2 with $j = 1$ then applies in the same manner and gives $p_{N-1} \le d/q$. The same argument applied to any nonzero $f \in W_\ell$ gives $p_\ell \le d/q$ for every $\ell \le N-1$, with no assumption.
\end{proof}

\paragraph*{Independence probability.}
By the chain rule of conditional probabilities:
\[
    \Pr[v_1,\dots,v_N\text{ lin.\ ind.}]
    \;=\;\prod_{\ell=0}^{N-1}\Pr\left[v_{\ell+1}\notin V_\ell \mid v_1,\dots,v_\ell \text{ lin.\ ind.}\right],
\]
with the $\ell=0$ factor $\Pr[v_1\neq 0]=1-q^{-n}$ (here $V_0=\{0\}$). Each factor is an average of $1-p_\ell$ over the admissible configurations $a_1,\dots,a_\ell$, and not a function of $\dim V_\ell$ alone; it is precisely because the bounds of Proposition~\ref{prop:per-step} are uniform over all such configurations that they may be substituted inside the product, each factor being at least $1-{(d/q)}^{r_\ell}$. Now $r_\ell=n-1$ for at most $N$ indices, while for the remaining ones $r_\ell=N-\ell=:j$ runs over $j=n-2,\dots,1$. Since $\prod_\ell(1-t_\ell)\ge 1-\sum_\ell t_\ell$ (valid for $t_\ell\in[0,1]$), and isolating the single term $j=1$,
\begin{equation}\label{eq:overall}
    \Pr[v_1,\dots,v_N\text{ lin.\ ind.}]\;\ge\;
    1-{\frac{d}{q}}
    -\sum_{j=2}^{n-2}{\left(\frac{d}{q}\right)}^{\!j}
    -N{\left(\frac{d}{q}\right)}^{\!n-1}
    -\frac{1}{q^{n}} .
\end{equation}

The third term involves $N$, which grows rapidly with $n$. The following bound is what allows it to be absorbed into ${(d/q)}^{2}$ in the proof of Theorem~\ref{thm:main}.

\begin{lemma}\label{lem:N-vs-3}
Let $d\ge3$ and $n\ge d^2+1$. Then $N=\binom{n+d-1}{d}\le 3^{\,n-3}$.
\end{lemma}

\begin{proof}
We first record that $d(d+1)\le 3^{d}/2$ for $d\ge3$. Indeed, keeping the first four terms of $3^d={(1+2)}^d$,
\[
    3^{d}\;\ge\;1+2d+4\binom d2+8\binom d3\;=\;2d^{2}+1+\tfrac43 d(d-1)(d-2),
\]
and $2d^2+1+\tfrac43d(d-1)(d-2)\ge 2d(d+1)$ reduces to $\tfrac43d(d-1)(d-2)\ge 2d-1$, which holds for $d\ge3$ since the left-hand side is $8$ at $d=3$ and grows cubically.
 
We now prove the statement by induction on $n$. For the base case $n=d^2+1$ we must check $\binom{d^2+d}{d}\le3^{\,d^2-2}$. From $\binom md\le m^d/d!$ with $m=d(d+1)$ and the inequality just proved,
\[
    \binom{d^2+d}{d}\;\le\;\frac{{\bigl(d(d+1)\bigr)}^{d}}{d!}\;\le\;\frac{3^{\,d^2}}{2^{d}\,d!}\;\le\;3^{\,d^2-2},
\]
 where the last inequality holds since $2^{d}d!\ge 48\ge 9$ for $d\ge3$. For the inductive step, $\binom{n+d}{d}=\frac{n+d}{n}\binom{n+d-1}{d}$ and $\frac{n+d}{n}\le3$, so $\binom{n+d}{d}\le3\cdot3^{\,n-3}=3^{\,(n+1)-3}$.
\end{proof}

\begin{theorem}\label{thm:main}
Let $d>2$, $n\ge d^2+1$ and $q\ge 3d$. Under Assumption~\ref{ass:generic},
\[
    \Pr[v_1,\dots,v_N\text{ lin.\ ind.}]\;\ge\;
    1-\frac{d}{q}-3{\left(\frac{d}{q}\right)}^{2}\;=\;1-O\!\left(\frac dq\right),
\]
the leading term $d/q$ coming from the step $\ell=N-1$, where the assumption is not used.
\end{theorem}

\begin{proof}
We estimate the three middle terms of~\Cref{eq:overall}. From $q\ge3d$ we get $d/q\le\frac13$, hence $1-d/q\ge\frac23$ and
\[
    \sum_{j=2}^{n-2}{\left(\frac dq\right)}^{\!j}\;\le\;\frac{{(d/q)}^{2}}{1-d/q}
    \;\le\;\frac32{\left(\frac dq\right)}^{2}.
\]
By Lemma~\ref{lem:N-vs-3} and $q/d\ge3$ we have $N\le3^{\,n-3}\le{(q/d)}^{\,n-3}$, so
\[
    N{\left(\frac dq\right)}^{n-1}\;\le\;{\left(\frac qd\right)}^{n-3}{\left(\frac dq\right)}^{n-1}
    \;=\;{\left(\frac dq\right)}^{2}.
\]
Finally $q^{-n}\le\frac12{(d/q)}^{2}$, since $d\ge3$ and $n\ge3$ give ${(d/q)}^{2}\ge 9q^{-2}\ge 2q^{-n}$. Hence,
\begin{align*}
    \Pr[v_1,\dots,v_N\text{ lin.\ ind.}]\;\ge\;&
    1-\frac{d}{q}
    -\sum_{j=2}^{n-2}{\left(\frac{d}{q}\right)}^{\!j}
    -N{\left(\frac{d}{q}\right)}^{\!n-1}
    -\frac{1}{q^{n}}\\
    \ge\,& 1-\frac{d}{q}
    -\frac{3}{2}{\left(\frac{d}{q}\right)}^{2}
    -{\left(\frac{d}{q}\right)}^{2}
    -\frac{1}{2}{\left(\frac{d}{q}\right)}^{2}\\
    =\;& 1-\frac{d}{q}-3{\left(\frac{d}{q}\right)}^{2}.
\end{align*}
\end{proof}

\paragraph*{Experimental results.} 
At cryptographic parameters the bound is close to $1$. Taking those of Kyber1024~\cite{NISTPQC:CRYSTALS-KYBER22}, whose module of rank $4$ over a ring of degree $256$ flattens to an LWE instance of dimension $n = 1024$ over $\F_q$ with $q = 3329$, and $d = 3$, Theorem~\ref{thm:main} gives a failure probability of at most $9.03\cdot10^{-4}$, i.e.\ a success probability of at least $0.99910$. These dimensions are far too large to test directly, so we validate the bound on the largest parameters that remain feasible.

We sampled $a_1,\dots,a_N$ uniformly from $\F_q^n$ and computed the rank of the $N\times N$ matrix with rows $\nu_d(a_1),\dots,\nu_d(a_N)$. Tables~\ref{tab:experiments-d3} and~\ref{tab:experiments-d4} report the observed frequency of full rank against the bound of Theorem~\ref{thm:main}, for $d=3$ over $10^{4}$ trials per parameter set and for $d=4$ over $10^{3}$, the latter being limited by the cost of a single trial at $N=4845$. All parameter sets satisfy $n\ge d^2+1$ and $q\ge 3d$; the first condition puts $d\ge5$ out of reach, since it would already require $N=\binom{30}{5}=142\,506$. All computations were performed using Magma V2.28-23~\cite{MR1484478} on a machine equipped with Intel$^\text{\textregistered}$ Core$^\text{\texttrademark}$ i9-14900KF CPU 32-cores and 126 GB of RAM.

\begin{table}[ht]
\centering
\setlength{\tabcolsep}{10pt}
\renewcommand{\arraystretch}{0.8}
\begin{tabular}{rrrrcc}
\toprule
$n$ & $d$ & $q$ & $N$ & Experimental & Bound (Thm.~\ref{thm:main}) \\
\midrule
10 & 3 & 17    & 220  & $0.9389$ & $0.7301$ \\
10 & 3 & 31    & 220  & $0.9648$ & $0.8751$ \\
10 & 3 & 97    & 220  & $0.9905$ & $0.9662$ \\
10 & 3 & 251   & 220  & $0.9960$ & $0.9876$ \\
10 & 3 & 1021  & 220  & $0.9995$ & $0.9970$ \\
10 & 3 & 3329  & 220  & $0.9997$ & $0.9991$ \\
\midrule
12 & 3 & 17    & 364  & $0.9424$ & $0.7301$ \\
12 & 3 & 31    & 364  & $0.9659$ & $0.8751$ \\
12 & 3 & 97    & 364  & $0.9908$ & $0.9662$ \\
12 & 3 & 251   & 364  & $0.9953$ & $0.9876$ \\
12 & 3 & 1021  & 364  & $0.9991$ & $0.9970$ \\
12 & 3 & 3329  & 364  & $0.9998$ & $0.9991$ \\
\midrule
16 & 3 & 17    & 816  & $0.9351$ & $0.7301$ \\
16 & 3 & 31    & 816  & $0.9669$ & $0.8751$ \\
16 & 3 & 97    & 816  & $0.9901$ & $0.9662$ \\
16 & 3 & 251   & 816  & $0.9964$ & $0.9876$ \\
16 & 3 & 1021  & 816  & $0.9992$ & $0.9970$ \\
16 & 3 & 3329  & 816  & $0.9997$ & $0.9991$ \\
\midrule
20 & 3 & 17    & 1540 & $0.9351$ & $0.7301$ \\
20 & 3 & 31    & 1540 & $0.9677$ & $0.8751$ \\
20 & 3 & 97    & 1540 & $0.9892$ & $0.9662$ \\
20 & 3 & 251   & 1540 & $0.9959$ & $0.9876$ \\
20 & 3 & 1021  & 1540 & $0.9992$ & $0.9970$ \\
20 & 3 & 3329  & 1540 & $0.9993$ & $0.9991$ \\
\bottomrule
\end{tabular}
\caption{Degree $d=3$, over $10^{4}$ trials per parameter set: observed frequency with which $\nu_d(a_1),\dots,\nu_d(a_N)$ are linearly independent, against the lower bound $1-d/q - 3{(d/q)}^{2}$ of Theorem~\ref{thm:main} (which holds under Assumption~\ref{ass:generic}). Here $N=\binom{n+d-1}{d}$ is the number of samples drawn in each trial. All parameter sets satisfy $n\ge d^2+1$ and $q\ge 3d$.}\label{tab:experiments-d3}
\end{table}

\begin{table}[ht]
\centering
\setlength{\tabcolsep}{10pt}
\renewcommand{\arraystretch}{0.8}
\begin{tabular}{rrrrcc}
\toprule
$n$ & $d$ & $q$ & $N$ & Experimental & Bound (Thm.~\ref{thm:main}) \\
\midrule
17 & 4 & 17    & 4845 & $0.9430$ & $0.5985$ \\
17 & 4 & 31    & 4845 & $0.9670$ & $0.8210$ \\
17 & 4 & 97    & 4845 & $0.9940$ & $0.9536$ \\
17 & 4 & 251   & 4845 & $0.9940$ & $0.9833$ \\
17 & 4 & 1021  & 4845 & $0.9980$ & $0.9960$ \\
17 & 4 & 3329  & 4845 & $1.0000$ & $0.9988$ \\
\bottomrule
\end{tabular}
\caption{Degree $d=4$ at the smallest admissible dimension $n=d^2+1=17$, over $10^{3}$ trials per parameter set. Columns as in Table~\ref{tab:experiments-d3}.}\label{tab:experiments-d4}
\end{table}

\subsection{Subsequent iterations}
After the first iteration, the polynomials undergo \texttt{Sgen} and \texttt{Sred} steps, which mix the coefficients of the newly generated S-polynomials, so that the Veronese structure of the coefficients is lost. 

To estimate the probability that the coefficient matrix of the homogeneous part of degree $i$ is diagonalizable, we therefore model it as a matrix $M_i$ with i.i.d.\ uniform rows in $\mathbb{F}_q^{m_i}$, where $m_i=\binom{n+i}{i}$. Although this is a heuristic approach to bound the success probability, our final (lower) bound (in \Cref{cor:easy2}) based on this modeling is consistent with the experimental success probability. In fact, the emperical success probability is
always greater (see \Cref{tab:d3psucc}) than our theoretical bound. We stress that this is a technique for bounding the probability, the correctness of \Cref{alg:alg} does not depend on it. A more precise (stochastic) modelling of the subsequent iterations is left as an interesting open problem.  

\begin{proposition}
Let $k_i=\binom{n+i-1}{i}$. The probability that $k_i$ uniform vectors of $\mathbb{F}_q^{k_i}$ are linearly independent is the classical
\[
\mathbb{P}_i=\prod_{j=0}^{k_i-1}\left(1-q^{\,j-k_i}\right)> 1-\frac{1}{q-1}.
\]
\end{proposition}

\begin{proof}
Let us consider one row at a time. Row $j+1$ preserves independence iff it avoids the span of the previous $j$ rows, a subspace of size $q^{\,j}$. This happens with conditional probability $1-q^{\,j-k_i}$, and multiplying over $j=0,\dots,k_i-1$ gives the formula. Using $\prod(1-a_j)\ge 1-\sum a_j$ (valid for $a_j\in[0,1]$),
\[
\mathbb{P}_i \;\ge\; 1-\sum_{j=0}^{k_i-1}q^{\,j-k_i} \;=\; 1-\sum_{t=1}^{k_i}q^{-t} \;>\; 1-\frac{1}{q-1}.
\]
\end{proof}

\subsection{Overall success probability}
Let $\mathbb{P}_i$ denote the probability of having full-rank at the $i$-th iteration, for $1\le i\le d$, and $F_i$ the event that iteration $i$ fails to have full rank, so $\mathbb{P}[F_i]=1-\mathbb{P}_i$. Using the union bound on the failure events gives
\[
p_{\text{succ}} = 1-\mathbb{P}\left[\bigcup_i F_i\right] \geq 1-\sum_{i=1}^{d}\bigl(1-\mathbb{P}_i\bigr) =\ \mathbb{P}_d-\sum_{i=1}^{d-1}\bigl(1-\mathbb{P}_i\bigr).
\]

\begin{proposition}[Success probability of \Cref{alg:alg} for $d>2$]\label{cor:easy2}
Let $d>2$ and $n\ge d^2+1$. Under Assumption~\ref{ass:generic} and for $q\ge 3d$,
\begin{enumerate}
    \item $p_{\text{succ}} \;\ge\; 1-\dfrac{2d}{q}-3{\left(\dfrac{d}{q}\right)}^{2}
    \;\ge\; 1-\dfrac{3d}{q}$;
    \item $p_{\text{succ}} \;\ge\; 1-\dfrac{5d}{2q}, \quad\text{ if moreover } q\ge 6d$.
\end{enumerate}
\end{proposition}

\begin{proof}
By Theorem~\ref{thm:main}, $\mathbb{P}_d\ge 1-\frac{d}{q}-3{(\frac{d}{q})}^2$, and by the previous proposition $1-\mathbb{P}_i<\frac1{q-1}$ for every $i<d$. Hence
\[
    p_{\text{succ}} \;\ge\; \mathbb{P}_d - \sum_{i=1}^{d-1}\bigl(1-\mathbb{P}_i\bigr)
    \;\ge\; 1-\frac{d}{q}-3{\left(\frac dq\right)}^{2} - \frac{d-1}{q-1}
    \;\ge\; 1-\frac{2d}{q}-3{\left(\frac dq\right)}^{2},
\]
where the last step uses $d\le q$ and $\frac{d}{q} \ge\frac{d-1}{q-1}$.
We ask for which $q$ this is at least the target bound.
\begin{enumerate}
    \item $1-\frac{2d}{q}-3{(\frac{d}{q})}^{2}\;\ge\;1-\frac{3d}{q} \iff 3{(\frac dq)}^{2}\le\frac{d}{q} \iff q\ge 3d$.
    \item $1-\frac{2d}{q}-3{(\frac dq)}^{2}\;\ge\;1-\frac{5d}{2q} \iff 3{(\frac dq)}^{2}\le\frac{d}{2q} \iff q\ge 6d$.
\end{enumerate}
\end{proof}

\paragraph{Experimental verification.}
The bound provided under Assumption~\ref{ass:generic} is also consistent with the experimental results, obtained on the same machine as the previous experiments. The observed success probability is taken over $10^4$ independent trials, and~\Cref{tab:d3psucc} shows that the bound of~\Cref{cor:easy2} is asymptotically high.

\begin{table}[H]
\centering
\setlength{\tabcolsep}{10pt}
\renewcommand{\arraystretch}{0.8}
\begin{tabular}{rrrr|cc|cc}
\toprule
$n$ & $d$ & $q$ & $N$ & \makecell{Exp.\\(1st it.)} & \makecell{Bound\\(Thm.~\ref{thm:main})} & \makecell{Exp.\\($p_{\text{succ}}$)} & \makecell{Bound\\(Prop.~\ref{cor:easy2})}\\
\midrule
10 & 3 & 19    & 220  & $0.9482$ & $0.7673$ & $0.8506$ & $0.6053$ \\
10 & 3 & 31    & 220  & $0.9648$ & $0.8751$ & $0.9090$ & $0.7581$ \\
10 & 3 & 97    & 220  & $0.9905$ & $0.9662$ & $0.9682$ & $0.9227$ \\
10 & 3 & 251   & 220  & $0.9960$ & $0.9876$ & $0.9886$ & $0.9701$ \\
10 & 3 & 1021  & 220  & $0.9995$ & $0.9970$ & $0.9972$ & $0.9927$ \\
10 & 3 & 3329  & 220  & $0.9997$ & $0.9991$ & $0.9986$ & $0.9977$ \\
\midrule
12 & 3 & 19    & 364  & $0.9431$ & $0.7673$ & $0.8406$ & $0.6053$ \\
12 & 3 & 31    & 364  & $0.9659$ & $0.8751$ & $0.9044$ & $0.7581$ \\
12 & 3 & 97    & 364  & $0.9908$ & $0.9662$ & $0.9688$ & $0.9227$ \\
12 & 3 & 251   & 364  & $0.9953$ & $0.9876$ & $0.9873$ & $0.9701$ \\
12 & 3 & 1021  & 364  & $0.9991$ & $0.9970$ & $0.9969$ & $0.9927$ \\
12 & 3 & 3329  & 364  & $0.9998$ & $0.9991$ & $0.9992$ & $0.9977$ \\
\midrule
16 & 3 & 19    & 816  & $0.9458$ & $0.7673$ & $0.8498$ & $0.6053$ \\
16 & 3 & 31    & 816  & $0.9669$ & $0.8751$ & $0.9047$ & $0.7581$ \\
16 & 3 & 97    & 816  & $0.9901$ & $0.9662$ & $0.9677$ & $0.9227$ \\
16 & 3 & 251   & 816  & $0.9964$ & $0.9876$ & $0.9875$ & $0.9701$ \\
16 & 3 & 1021  & 816  & $0.9992$ & $0.9970$ & $0.9975$ & $0.9927$ \\
16 & 3 & 3329  & 816  & $0.9997$ & $0.9991$ & $0.9988$ & $0.9977$ \\
\midrule
20 & 3 & 19    & 1540 & $0.9427$ & $0.7673$ & $0.8507$ & $0.6053$ \\
20 & 3 & 31    & 1540 & $0.9677$ & $0.8751$ & $0.9040$ & $0.7581$ \\
20 & 3 & 97    & 1540 & $0.9892$ & $0.9662$ & $0.9681$ & $0.9227$ \\
20 & 3 & 251   & 1540 & $0.9959$ & $0.9876$ & $0.9886$ & $0.9701$ \\
20 & 3 & 1021  & 1540 & $0.9992$ & $0.9970$ & $0.9973$ & $0.9927$ \\
20 & 3 & 3329  & 1540 & $0.9993$ & $0.9991$ & $0.9991$ & $0.9977$ \\
\bottomrule
\end{tabular}
\caption{Degree $d=3$, over $10^{4}$ trials per parameter set. Exp.\ (1st it.) is the observed frequency with which the coefficient matrix of the first iteration is diagonalizable, against the bound $1-d/q-3{(d/q)}^{2}$ of Theorem~\ref{thm:main}. Exp.\ ($p_{\text{succ}}$) is the observed success probability of \Cref{alg:alg}, against the bound $1-5d/2q$ of~\Cref{cor:easy2}. Both hold under Assumption~\ref{ass:generic}. Here $N=\binom{n+d-1}{d}$ is the number of samples drawn in each trial, and all parameter sets satisfy $n\ge d^2+1$ and $q\ge 6d$.}\label{tab:d3psucc}
\end{table}

\section{Conclusions}\label{sec:conclusions}
Our proposed algebraic algorithm is probabilistic and shows polynomial improvement in time complexity over existing algebraic algorithms for LWE. Moreover, our algorithm also applies to the Regular Syndrome Decoding (RSD) problem: under the algebraic modelling of Briaud and {\O}ygarden~\cite{DBLP:conf/eurocrypt/BriaudO23}, RSD reduces to an overdetermined system of low-degree polynomials sharing an algebraic structure similar to the one arising from LWE (with bounded errors), a connection already exploited in~\cite{EC:NMSU25}.

Although Assumption~\ref{ass:generic} was necessary to bound the success probability for the case $d > 2$, our experimental results indicate that this assumption holds with high probability. This observation leads to two interesting open questions: whether the probability of this assumption can be precisely estimated or bounded. Alternatively, can the success probability be bounded without such assumption? 

Furthermore, our technique for generating independent S-polynomials leads to a notable observation: during the S-polynomial generation process, substantially more candidate pairs are available than are actually utilized. Thus, another open question remains: whether the sample size can be reduced while maintaining the complexity improvement over existing algorithms.

\section*{Acknowledgments}
We thank anonymous reviewers whose comments have helped us to improve the technical quality of the paper.
Part of this work appears in Massimo Fumiani's master's thesis~\cite{Fum26}.

\bibliographystyle{splncs04}
\bibliography{main}

\appendix
\section{Next combination in Gray code order}\label{sec:successor}

\begin{algorithm}[t!]
\caption{Next combination in Gray order}\label{alg:successor}
\begin{algorithmic}[1]
\Statex{\hspace*{-\algorithmicindent} \textbf{Input:} $a=(a_1,\dots,a_{n})$: $n$-tuple, $n>1$, total sum $d<n$}
\Statex{\hspace*{-\algorithmicindent} \textbf{Output:} $b=(b_1,\dots,b_{n})$ next $n$-tuple combination in Gray order, or $\varnothing$}
\Statex{\hspace*{-\algorithmicindent} \textbf{Uses:} 
\texttt{call} (simulate recursive call), 
\texttt{after} (action after simulated recursion), 
\texttt{next}/\texttt{prev} (direction to scan the current component; 
$\overline{\texttt{next}}$ denotes \texttt{prev} and vice versa)}

\Function{ComputeTail}{$m$, $s$, parity}
  \State{\IfThen{$m \le 0$}{\Return{$\varnothing$}}}
  \State{\IfThenElse{parity = 0}{\Return{$(\mathbf{0}^{m-1}) || (s)$}}{\Return{$(s) || (\mathbf{0}^{m-1})$}}}
\EndFunction{}
\Function{Successor}{$a$, $n$, $d$}
\State{\IfThen{$a = \varnothing$}{\Return{$(\mathbf{0}^{n-1}, d)$}}}\Comment{First combination in Gray code order}

\State{$P_k=\sum_{j<k} a_j$, for $1\le k\le n$}\Comment{Compute prefix sums}
\State{\texttt{stack} $\gets$ $(\texttt{call},\texttt{next},0,-1)$}\Comment{Initialize stack}
\State{\texttt{last\_ret} $\gets \varnothing$}

\While{\texttt{stack} not empty}

    \State{$(type,mode,k,i)$ $\gets$ \texttt{stack}.pop()}

    \If{$type=\texttt{call}$}
        \If{$k=n-1$}
            \State{\texttt{last\_ret} $\gets \varnothing$}
        \Else{}
            \State{$i \gets a_k$}
            \State{\texttt{stack}.push$((\texttt{after},mode,k,i))$}
            \State{$mode$ $\gets$ \GetsIfThenElse{$mode$}{$i \bmod 2 = 0$}{$\overline{mode}$}}
            \State{\texttt{stack}.push$((\texttt{call},mode,k+1, -1))$ }\Comment{Simulate recursive call at level $k+1$}
        \EndIf{}

    \Else{}
        \State{\texttt{rem} $\gets d-P_k$}

        \If{\texttt{last\_ret} $\neq \varnothing$}
            \State{\texttt{last\_ret} $\gets (i)\ ||\ \texttt{last\_ret}$}
        \Else{}
            \State{\IfThenElse{$mode=\texttt{next}$}{$i \gets i+1$}{$i \gets i-1$}}
            \If{$0 \le i \le \texttt{rem}$}
                \State{\texttt{tail} $\gets$ \Call{ComputeTail}{$n-k-1$, $\texttt{rem}-i$, $\gets i \bmod 2$}}
                \State{\texttt{last\_ret} $\gets (i)\ ||\ \texttt{tail}$}
            \EndIf{}
        \EndIf{}
    \EndIf{}

\EndWhile{}

\State{\Return{\texttt{last\_ret}}}
\EndFunction{}
\end{algorithmic}
\end{algorithm}

Our construction follows the general principles underlying Gray code generation for multiset combinations as developed in~\cite{DBLP:journals/ejc/RuskeyS96,DBLP:journals/cj/Takaoka99,DBLP:journals/corr/Takaoka15a}, which we adapt to the requirements of our setting. 
In particular, \Cref{alg:successor} defines the \textsc{Successor} procedure, which, given an input tuple, returns the next tuple in the Gray code ordering (equivalently, the next vertex along the corresponding Hamiltonian path, as visually exemplified in \Cref{fig:graycode4vars}). 
This allows admissible pairs to be generated directly by traversing the Gray code ordering, while preserving the efficiency of the underlying construction. 
It has been proved that the successor of a combination can be computed in constant time, and our adaptation retains this property~\cite{DBLP:journals/cj/Takaoka99,DBLP:journals/corr/Takaoka15a}.
Since the procedure is a direct modification of the constructions presented in the cited works, we do not address its correctness here, as it follows from the same arguments.

For clarity, \Cref{alg:successor} uses the stack labels \texttt{call} and \texttt{after} to simulate recursion, and \texttt{next}/\texttt{prev} to indicate the direction of the component scan.
The procedure internally calls \textsc{ComputeTail} to generate the tail of the tuple (i.e., the zero-padded entries) according to the remaining sum and parity.  
For further details, we refer the reader to Algorithm 2 and Section 7 in~\cite{DBLP:journals/corr/Takaoka15a}.

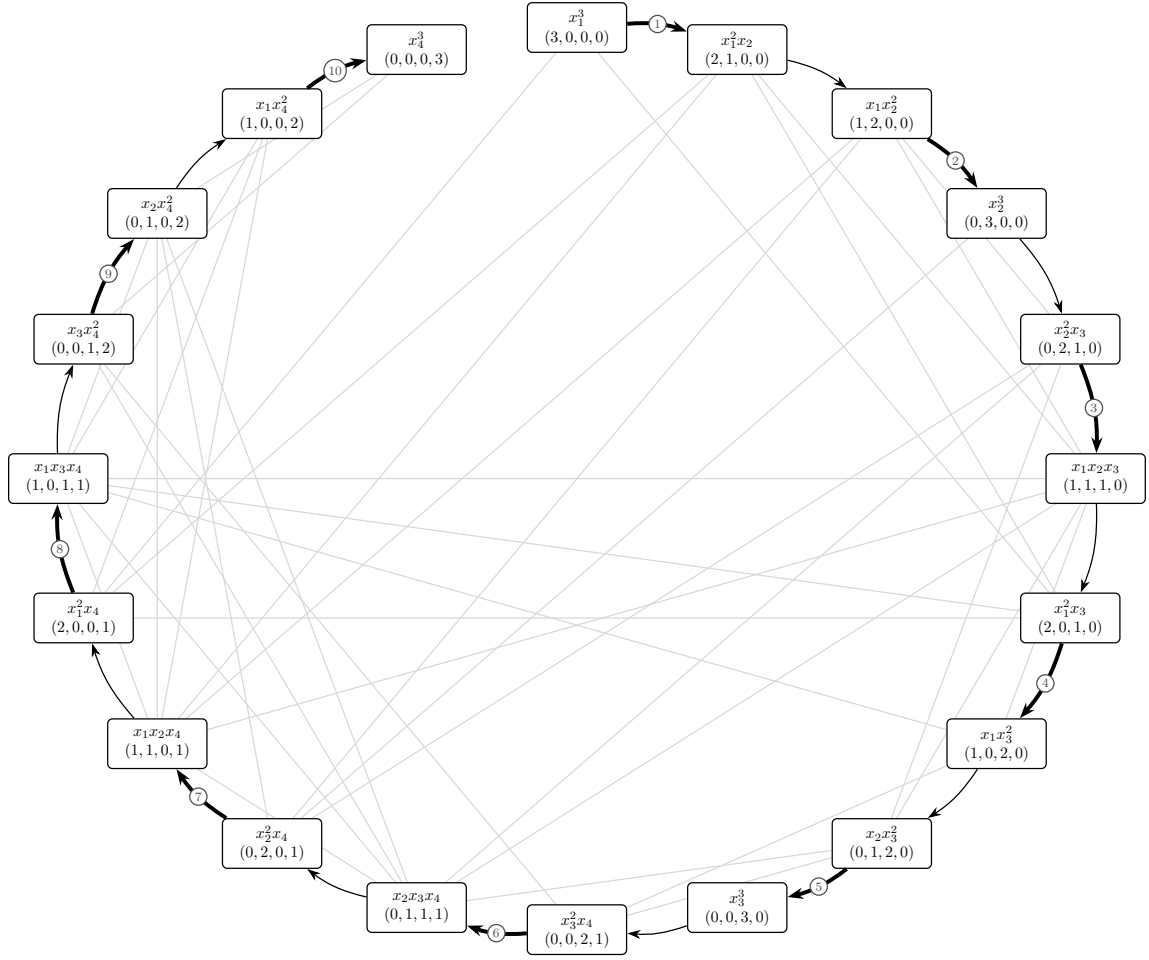
\begin{figure}[h!]
\centering
\begin{tabular}{cc}
  \resizebox{0.95\textwidth}{!}{\begin{tikzpicture}[
        monomial/.style={rectangle, rounded corners=3pt, draw=black, fill=white, align=center, inner sep=2pt, minimum width=2.2cm, minimum height=1.1cm, font=\fontsize{10}{12}\selectfont, thick},
        pairarrow/.style={-{Stealth[length=3.5mm, width=2.5mm]}, line width=2.5pt, black},
        transition/.style={-{Stealth[length=3mm, width=2mm]}, thick, black},
        standard/.style={draw=lipicsLightGray, line width=0.8pt}
    ]
    \def\Rx{11.5} 
    \def\Ry{10} 

    \node[monomial] (N1)  at (90:\Rx\space and \Ry)   {$x_1^3$ \\ $(3,0,0,0)$};
    \node[monomial] (N2)  at (72:\Rx\space and \Ry)   {$x_1^2 x_2$ \\ $(2,1,0,0)$};
    \node[monomial] (N3)  at (54:\Rx\space and \Ry)   {$x_1 x_2^2$ \\ $(1,2,0,0)$};
    \node[monomial] (N4)  at (36:\Rx\space and \Ry)   {$x_2^3$ \\ $(0,3,0,0)$};
    \node[monomial] (N5)  at (18:\Rx\space and \Ry)   {$x_2^2 x_3$ \\ $(0,2,1,0)$};
    \node[monomial] (N6)  at (0:\Rx\space and \Ry)    {$x_1 x_2 x_3$ \\ $(1,1,1,0)$};
    \node[monomial] (N7)  at (-18:\Rx\space and \Ry)  {$x_1^2 x_3$ \\ $(2,0,1,0)$};
    \node[monomial] (N8)  at (-36:\Rx\space and \Ry)  {$x_1 x_3^2$ \\ $(1,0,2,0)$};
    \node[monomial] (N9)  at (-54:\Rx\space and \Ry)  {$x_2 x_3^2$ \\ $(0,1,2,0)$};
    \node[monomial] (N10) at (-72:\Rx\space and \Ry)  {$x_3^3$ \\ $(0,0,3,0)$};
    \node[monomial] (N11) at (-90:\Rx\space and \Ry)  {$x_3^2 x_4$ \\ $(0,0,2,1)$};
    \node[monomial] (N12) at (-108:\Rx\space and \Ry) {$x_2 x_3 x_4$ \\ $(0,1,1,1)$};
    \node[monomial] (N13) at (-126:\Rx\space and \Ry) {$x_2^2 x_4$ \\ $(0,2,0,1)$};
    \node[monomial] (N14) at (-144:\Rx\space and \Ry) {$x_1 x_2 x_4$ \\ $(1,1,0,1)$};
    \node[monomial] (N15) at (-162:\Rx\space and \Ry) {$x_1^2 x_4$ \\ $(2,0,0,1)$};
    \node[monomial] (N16) at (-180:\Rx\space and \Ry) {$x_1 x_3 x_4$ \\ $(1,0,1,1)$};
    \node[monomial] (N17) at (-198:\Rx\space and \Ry) {$x_3 x_4^2$ \\ $(0,0,1,2)$};
    \node[monomial] (N18) at (-216:\Rx\space and \Ry) {$x_2 x_4^2$ \\ $(0,1,0,2)$};
    \node[monomial] (N19) at (-234:\Rx\space and \Ry) {$x_1 x_4^2$ \\ $(1,0,0,2)$};
    \node[monomial] (N20) at (-252:\Rx\space and \Ry) {$x_4^3$ \\ $(0,0,0,3)$};

    \begin{scope}[on background layer]
        \foreach \i/\j in {1/7,1/15,2/6,2/7,2/14,2/15,3/5,3/6,3/13,3/14,4/13,5/9,5/12,5/13,6/8,6/9,6/12,6/14,6/16,7/15,7/16,8/11,8/16,9/11,9/12,11/17,12/14,12/16,12/17,12/18,13/18,14/16,14/18,14/19,15/19,16/18,16/19,17/20,18/20} {
            \draw[standard] (N\i) -- (N\j);
        }
    \end{scope}

    \foreach \src/\dst/\num in {1/2/1,3/4/2,5/6/3,7/8/4,9/10/5,11/12/6,13/14/7,15/16/8,17/18/9,19/20/10} {
        \draw[pairarrow] (N\src) to[bend left=12] node[midway, fill=white, draw=darkgray, thick, circle, inner sep=1.5pt, font=\scriptsize, text=darkgray] {\num} (N\dst);
    }

    \foreach \src/\dst in {2/3,4/5,6/7,8/9,10/11,12/13,14/15,16/17,18/19} {
        \draw[transition] (N\src) to[bend left=12] (N\dst);
    }
\end{tikzpicture}} &
\end{tabular}
\caption{Graph with monomials for $n=4$ variables and degree $d=3$. Edges represent all available connections between monomials whose greatest common divisor has degree $d-1=3$. The directed path traces the Hamiltonian path derived via the Gray code, with numbered arrows highlighting the specific admissible pairs extracted.}\label{fig:graycode4vars}
\end{figure}
\section{Solving LWE with Binary Secrets}\label{sec:binarysecret}
In the \bslwe\ setting, the secret $\mathbf{s} = (s_1,\dots,s_n)$ satisfies $s_i \in \{0,1\}$, which implies the field equations $x_i^2 - x_i = 0$ for all $1 \le i \le n$. The error distribution remains unchanged.

To solve this system, we append the set $\{x_i^2 - x_i = 0 : 1 \le i \le n\}$ to the initial polynomial system and reduce all equations modulo the ideal generated by these relations. Consequently, each polynomial will consist strictly of square-free monomials of the form $\mathbf{x}^\alpha = x_1^{\alpha_1}\cdots x_n^{\alpha_n}$, where $\alpha_i \in \{0,1\}$. 
For polynomials of degree at most $d$, the total number of such monomials is $\sum_{i=0}^d \binom{n}{i}$, where the number of monomials of degree exactly $i$ is exactly $\binom{n}{i}$.

The core algorithm remains unchanged. The complexity analysis carries over naturally by substituting the general monomial count with the square-free equivalent. Because both $\binom{n+i-1}{i}$ and $\binom{n}{i}$ share the same dominant term $\mathcal{O}(n^i/i!)$ asymptotically, the bounds remain functionally identical to the general case.
The minimum number of samples required to initialize the matrix is $\binom{n}{d}$. Under this modification, the complexities of the main procedures evaluate to:
\begin{align*}
    \mathscr{C}_{\texttt{Sgen}} &= \mathcal{O}\left((d-1)\cdot \frac{n^{2(d-1)}}{{((d-1)!)}^2} \right), \\
    \mathscr{C}_{\texttt{Sred}} &= \mathcal{O}\left((d-1)\cdot \frac{n^{1 + \omega(d-1)}}{{((d-1)!)}^\omega} \right), \\
    \mathscr{C}_{\texttt{Diag}} &= \mathcal{O}\left(d\cdot \frac{n^{d\omega}}{{(d!)}^\omega} \right).
\end{align*}
Combining these bounds as in \Cref{thm:result}, the overall complexity of the algorithm in the \bslwe\ setting is $\mathcal{O}\left( d\cdot{n^{d\omega}}/{{(d!)}^\omega}  \right)$ field operations.

Furthermore, to ensure that the initial system contains a sufficient number of distinct polynomials to form exactly $\binom{n}{d-1}$ disjoint pairs, the dimensional condition $n \ge 3d - 1$ must be satisfied.

\begin{proposition}
An initial system of $\binom{n}{d}$ distinct square-free polynomials of degree $d$ contains $\binom{n}{d-1}$ disjoint pairs if and only if $n \ge 3d - 1$.
\end{proposition}

\begin{proof}
The statement follows directly by imposing the condition that the total number of initial polynomials must be at least twice the number of required pairs, namely $\binom{n}{d} \ge 2\binom{n}{d-1}$. Expanding the binomial coefficients yields the required bound on $n$.
\end{proof}
\section{Auxiliary definitions}\label{sec:dreg}
 
\subsection{Regular sequence}

\begin{definition}[Regular sequence]\label{def:regular}
Let $\mathbb{K}$ be a field and let $R = \mathbb{K}[x_1,\dots,x_n]$. Let $f_1,\dots,f_m \in R$ be homogeneous polynomials with $\deg f_i = d_i \geq 1$, and let $I = \langle f_1,\dots,f_m \rangle$. The sequence $(f_1,\dots,f_m)$ is said to be \emph{regular} if
\begin{enumerate}
  \item[$(i)$] $I \neq R$; and
  \item[$(ii)$] for every $i \in \{1,\dots,m\}$, the polynomial $f_i$ is a non-zero-divisor in the quotient ring $R / \langle f_1,\dots,f_{i-1} \rangle$, that is,
  \[
  \forall\, g \in R : \quad g\,f_i \in \langle f_1,\dots,f_{i-1} \rangle\;\Longrightarrow\; g \in \langle f_1,\dots,f_{i-1} \rangle ,
  \]
  with the convention $\langle f_1,\dots,f_{i-1}\rangle = \{0\}$ for $i = 1$.
\end{enumerate}
\end{definition}

\subsection{Degree of regularity}
The degree of regularity was introduced by Bardet, Faug{\`e}re, and Salvy in~\cite{bardet2004etude,faugere2004complexity}.

\begin{definition}[Degree of Regularity]\label{def:dreg}
    Let $\mathbb{K}$ be a field and $F = \{f_1, \dots, f_m\}\allowbreak \subseteq \mathbb{K}[x_1, \dots, x_n]$ ($m > n$) be a system of homogenous polynomials of degree $d > 0$. The degree of regularity of a zero dimensional ideal $I = \langle F \rangle$ is defined as
    \[
    d_{\text{reg}}(I) = \min \left\{ d \ge 0 : \dim_{\mathbb{K}}\left( \{f \in I, \deg(f) = d \} \right) = \binom{n+d-1}{d} \right\}.
    \]
\end{definition}
\end{document}